\documentclass[10pt,journal,compsoc]{IEEEtran}
\ifCLASSOPTIONcompsoc
  \usepackage[nocompress]{cite}
\else
  \usepackage{cite}
\fi

\usepackage{amsmath,amssymb,amsfonts,amssymb,amscd}
\usepackage{fancyhdr}
\usepackage{mathrsfs}
\usepackage{graphicx}
\usepackage{textcomp}
\usepackage[usenames,dvipsnames]{xcolor}

\usepackage{subcaption}
\usepackage{multirow}
\usepackage{amsthm}
\usepackage[T1]{fontenc}
\usepackage{enumitem}
\usepackage{algorithm,algorithmicx,algpseudocode}
\usepackage{listings}
\usepackage{pifont}
\usepackage{booktabs}
\usepackage[hidelinks]{hyperref}
\usepackage{enumitem}
\setlist[itemize]{leftmargin=*}
\setlist[enumerate]{leftmargin=*}
\setlist{nolistsep}
\newcommand{\one}{\ding{182}}
\newcommand{\two}{\ding{183}}
\newcommand{\three}{\ding{184}}
\newcommand{\four}{\ding{185}}

\newcommand{\macsection}[1]{\noindent\textbf{#1}~~~~}

\newtheorem{theorem}{Theorem}[section]

\newtheorem{lemma}[theorem]{Lemma}

\newtheorem{assumption}{Assumption}

\DeclareMathOperator{\E}{\mathbb{E}}

\def\BibTeX{{\rm B\kern-.05em{\sc i\kern-.025em b}\kern-.08em
    T\kern-.1667em\lower.7ex\hbox{E}\kern-.125emX}}

\ifCLASSINFOpdf
\else
\fi

\begin{document}

%
\title{Breaking (Global) Barriers in Parallel Stochastic Optimization with Wait-Avoiding Group Averaging}

\author{Shigang~Li,~\IEEEmembership{Member~IEEE,} 
        Tal~Ben-Nun,
        Giorgi~Nadiradze, 
        Salvatore~Di~Girolamo,
        Nikoli~Dryden,
        Dan~Alistarh,
        and~Torsten~Hoefler,~\IEEEmembership{Senior~Member~IEEE}
\IEEEcompsocitemizethanks{
\IEEEcompsocthanksitem Shigang Li is with Department of Computer Science, ETH Zurich. E-mail: shigangli.cs@gmail.com
\IEEEcompsocthanksitem Tal Ben-Nun is with Department of Computer Science, ETH Zurich. E-mail: talbn@inf.ethz.ch
\IEEEcompsocthanksitem Giorgi Nadiradze is with IST Austria. E-mail: giorgi.nadiradze@ist.ac.at
\IEEEcompsocthanksitem Salvatore Di Girolamo is with Department of Computer Science, ETH Zurich. E-mail: salvatore.digirolamo@inf.ethz.ch
\IEEEcompsocthanksitem Nikoli Dryden is with Department of Computer Science, ETH Zurich. E-mail: ndryden@ethz.ch
\IEEEcompsocthanksitem Dan Alistarh is with IST Austria. E-mail: dan.alistarh@ist.ac.at
\IEEEcompsocthanksitem Torsten Hoefler is with Department of Computer Science, ETH Zurich. E-mail: htor@inf.ethz.ch}
\thanks{(Corresponding Author: Shigang Li.) \\ Published in IEEE Transactions on Parallel and Distributed Systems (TPDS), Vol. 32, No. 7, July 2021, DOI: \textcolor{Blue}{\url{https://doi.org/10.1109/TPDS.2020.3040606}}}
}

\IEEEtitleabstractindextext{%
\begin{abstract}
Deep learning at scale is dominated by communication time. Distributing samples across nodes usually yields the best performance, but poses scaling challenges due to global information dissemination and load imbalance across uneven sample lengths. State-of-the-art decentralized optimizers mitigate the problem, but require more iterations to achieve the same accuracy as their globally-communicating counterparts. We present Wait-Avoiding Group Model Averaging (WAGMA) SGD, a wait-avoiding stochastic optimizer that reduces global communication via subgroup weight exchange. The key insight is a combination of algorithmic changes to the averaging scheme and the use of a group allreduce operation. We prove the convergence of WAGMA-SGD, and empirically show that it retains convergence rates similar to Allreduce-SGD. For evaluation, we train ResNet-50 on ImageNet; Transformer for machine translation; and deep reinforcement learning for navigation at scale. Compared with state-of-the-art decentralized SGD variants, WAGMA-SGD significantly improves training throughput (e.g., 2.1$\times$ on 1,024 GPUs for reinforcement learning), and achieves the fastest time-to-solution (e.g., the highest score using the shortest training time for Transformer).
\end{abstract}

}

\maketitle

\IEEEdisplaynontitleabstractindextext

%
\IEEEpeerreviewmaketitle

\section{Introduction}

The introduction of deep learning is one of the most important advancements in science over the past two decades, powering industries from autonomous driving~\cite{bojarski2016end} to drug discovery~\cite{alphafold}. With the rise of deep neural networks, their training evolved into a computationally-intensive task that consumes as many resources as modern complex high-performance computing problems~\cite{openai}. As a result, an abundance of research has been conducted into its scaling and distribution~\cite{dl-survey}.

The leading contenders for largest workloads in deep learning are Neural Language Models~\cite{kaplan2020scaling,gpt-2}, Deep Reinforcement Learning (RL)~\cite{alphastar,mcc2018empirical} and Neural Architecture Search~\cite{liu2018darts}. In these regimes, computation time is measured in thousands of ``GPU days'', with some utilizing hundreds of accelerators (GPUs, TPUs) for several weeks~\cite{alphastar,sparcml,regevo}. 

Distributed training is largely supported by data parallelism, where sample evaluation is partitioned across processors. In this parallelism mode, all participants must exchange their gradients or model, resulting in an \texttt{Allreduce} operation across a cluster~\cite{mpi-3.1}. The exchange communication dominates the overall runtime~\cite{sparcml}, especially in large-minibatch stochastic gradient descent (SGD). To exacerbate the problem, certain datasets and environments are inherently imbalanced, e.g., with different sentence/video lengths~\cite{li2020taming} or heterogeneous environments in RL~\cite{wijmans2019decentralized}. 

In order to mitigate the wait time for gradient/weight exchange, existing approaches attempt to relax model consistency between processors~\cite{dl-survey,tang2020}. Examples include synchronous gossip-based SGD~\cite{syncring,assran2018stochastic}, asynchronous SGD~\cite{hogwild,mnih2016asynchronous,asyncring,gossipgrad}, and asynchronous SGD with bounded staleness~\cite{ssp13,zhang2015deep,chen2016scalable,stich2018local}. Gossip-based SGD replaces the global allreduce by communicating with randomly selected neighbors. Asynchronous SGD breaks the global synchronization to mitigate
the effect of stragglers (slow processes). However, most of these approaches adversely impact convergence, necessitating an increase in the number of iterations~\cite{assran2018stochastic, nadiradze2019popsgd}, sometimes to the point where synchronous waits are preferable.

In this paper, we solve this problem by introducing Wait-Avoiding Group Model Averaging (WAGMA) SGD, a novel optimizer that combines group collective communication with bounded staleness, in order to ensure competitive performance with decentralized and asynchronous methods, while retaining the convergence rate of synchronous model-averaging SGD. 
WAGMA-SGD locally communicates model updates across subgroups of processors, mitigating the need for global communication at every training iteration. 
Specifically, we propose to use a group allreduce operation for model averaging, in which the fastest process will trigger exchanges within all subgroups. Grouping is performed dynamically to facilitate
model update propagation, and as a result not only speeds up communication, but also mitigates the effect of unbalanced workloads, all without harming convergence in practice. 

We theoretically prove the convergence of WAGMA-SGD, showing that, for certain parameter values, its convergence rate is comparable to synchronous SGD with model averaging. 
Subsequently, we test the algorithm on a supercomputer equipped with GPUs for three different categories of deep learning: supervised image classification on the ImageNet dataset; semi-supervised language modeling on the WMT17 translation dataset; and deep reinforcement learning on the Habitat indoor navigation dataset.
We show that both theoretically and empirically, WAGMA-SGD is favorable over other asynchronous algorithms and the baselines, which makes it an excellent approach for scaling up distributed deep learning.

Our main contributions are:
\begin{itemize}
  \item We propose a novel asynchronous decentralized optimizer --- WAGMA-SGD, and realize it based on a wait-avoiding group allreduce operation.
  \item We theoretically analyze the convergence of WAGMA-SGD, showing that, under reasonable parameter values, it converges at the same rate as SGD, with linear speedup due to parallel updates. 
  \item Compared with state-of-the-art decentralized SGD, WAGMA-SGD significantly improves the training throughput (e.g., 2.1$\times$ on 1,024 GPUs on RL), and achieves the fastest time-to-solution for all three evaluated tasks (e.g., the highest score using the shortest training time for Transformer).
\end{itemize}

\section{Background and Related Work}
\label{sec:bg}

Deep neural networks are primarily trained with mini-batch SGD~\cite{bottou2018optimization}. Let $b$ be the batch size, $W_t$ the neural network weights at step $t$, $(x_i, y_i)$ a set of samples of size $b$, and $\ell$ a loss function. We compute the loss for each sample as $z_i = \ell(W_t, x_i, y_i)$ and then a stochastic gradient as
\[ G_t = \frac{1}{b} \sum_{i = 0}^b \nabla \ell(W_t, z_i). \]
SGD then iterates over steps such that $W_{t+1} = W_t - \eta_t G_t$. In more general terms, first-order stochastic gradient update rules can take different forms (e.g., by adding a momentum term), which is represented as $W_{t+1}=W_t + U\left(G_t,W_{(0,\dots,t)},t\right)$.  In distributed environments with $P$ processors, $b$ denotes the \textit{local} batch size per processor. We refer to Ben-Nun \& Hoefler~\cite{dl-survey} for a general overview of distributed deep learning.

Thanks to the robustness of stochastic optimization, in distributed environments one can relax weight updates by varying several axes, trading off communication overhead for convergence. Data-parallel distributed SGD algorithms can be broadly identified by the following five questions:

\macsection{Q1. What is averaged?}

There are two typical approaches for aggregating distributed updates: gradient and model averaging. When performing gradient averaging, we compute $G_t$ as the average over the global batch size. With standard model averaging, the SGD update is applied locally at the node, and then the resulting model $W_{t+1}$ is averaged over all processors.

Complementary to these approaches is the degree of quantization or sparsity in the exchanged updates. As these concepts are out of the scope of this paper, we refer to Tang et al.~\cite{tang2020} for a comprehensive survey.

\macsection{Q2. Who is coordinating the averaging?}

Earlier implementations of distributed SGD for deep learning~\cite{dean12} use a \textit{centralized} coordination architecture, where a parameter server or other coordinator maintains a master copy of the model that workers use. As this approach does not scale to large numbers of processors, a \textit{decentralized} global clock can be synchronized across workers, where each worker maintains a local replica of the model and communicates updates to other workers directly.

To mitigate the overheads of global communication and synchronization, several decentralized instances of SGD have been proposed, e.g.,~\cite{syncring, asyncring, assran2018stochastic, nadiradze2019popsgd}, where each worker maintains a local model but communicates updates in separate schedules, rather than synchronizing globally.

\macsection{Q3. How old (stale) can averaged components be?}

In a \emph{synchronous} system, model or gradient averaging occurs when all processes are on the same training iteration $t$. This does not guarantee that every worker uses the same parameters (i.e., consistent model), however, standard parameter server or globally-coordinated methods ensure all workers have a consistent model. In an \emph{asynchronous} system, averaging can occur between workers at any point. We thus define the staleness of models/gradients by $\tau$, indicating how many iterations have passed since the produced value's model was updated. A \emph{bounded staleness} system mitigates convergence issues with asynchronous systems by ensuring that the difference in the number of training iterations between the slowest and fastest processor is bounded, using $\tau$ as a proxy.

\macsection{Q4. How often is global averaging performed?}

While bounded and unbounded staleness SGD variants do not adhere to rigid communication schedules, some algorithms may periodically synchronize all processors' model replicas. This ensures that not only the staleness is bounded by $\tau$, but also the consistency of the model is retained throughout training, mitigating its divergence across processors. In other algorithms, this global consensus is achieved post-training, by choosing the model average or the model with best generalization scores. Note that under this nomenclature, the global average frequency of synchronous variants is one step.

\begin{table*}[ht!]
  \caption{Classification of data-parallel SGD variants.}
  \label{tab:rel-sgd}
  \centering
  \renewcommand{\arraystretch}{1.2}
  \setlength{\tabcolsep}{5pt}
  \begin{tabular}{llp{7cm}p{6.3cm}}
    \toprule
    Coordination & Staleness & Gradient Averaging &  Model Averaging  \\
    \midrule
    \multirow{4}{*}{Centralized} &
    None &
 Parameter server~\cite{scaling14}, P3~\cite{jayarajan2019priority} &
 --- \\
    & Unbounded &
Hogwild!~\cite{hogwild}, Downpour SGD~\cite{dean12}, AASGD~\cite{grishchenko2018asynchronous} & SAPS-PSGD~\cite{tang2020communication}
\\
    & Bounded &
SSP~\cite{ssp13}, Rudra~\cite{rudra}, Softsync SGD~\cite{zhang2015staleness}, Gaia~\cite{gaia}, $k$-async SGD~\cite{dutta2018slow}, Qsparse-local-SGD~\cite{basu2019qsparse}, Hybrid sync/async~\cite{kurth2017deep} & EASGD~\cite{zhang2015deep}, Federated learning~\cite{mcmahan2017communication,konevcny2016federated}
\\
    \midrule
    \multirow{3}{*}{\shortstack[l]{Decentralized,\\ $S=P$}} &
    None &
    \textbf{Allreduce-SGD}~\cite{gibiansky2017bringing,sergeev2018horovod,deep500} & BMUF~\cite{chen2016scalable}
    \\
    & Unbounded & ---
    & One-shot SGD~\cite{mcdonald2009efficient}, \ \ \ \ \ WP-SGD~\cite{cheng2020wp}, \ \ \ \ \ \ \ \ \ \ \ \ \ \ \ \ \ \ SimuParallelSGD~\cite{zinkevich2010parallelized}
    \\
    & Bounded & \textbf{Eager-SGD}~\cite{li2020taming}, Gradient lag~\cite{kurth2018exascale}
    & ---
    \\

 \midrule
\multirow{3}{*}{\shortstack[l]{Decentralized,\\ $S=\sqrt{P}$}} &
None & ---  & ---\\
& Unbounded & ---
& ---
\\
& Bounded & --- & $\bigstar$\textbf{WAGMA-SGD}$\bigstar$
\\

    \midrule
    \multirow{3}{*}{\shortstack[l]{Decentralized,\\ $S=\mathcal{O}(1)$}} &
    None &
--- & \textbf{D-PSGD}~\cite{syncring}, \textbf{SGP}~\cite{assran2018stochastic}
\\
    & Unbounded & GossipGraD~\cite{gossipgrad}, Choco-SGD~\cite{koloskova2019decentralized}
& \textbf{AD-PSGD}~\cite{asyncring}, Gossiping SGD~\cite{jin16},  SwarmSGD~\cite{nadiradze2019popsgd}
\\
    & Bounded & CDSGD~\cite{jiang2017collaborative}
 & \textbf{Local SGD}~\cite{stich2018local,lin2018don,wang2019adaptive}
\\
    \bottomrule
  \end{tabular}
\end{table*}

\macsection{Q5. How many learners are averaging at every step?}

In the steps between the aforementioned global model averaging period, decentralized SGD variants perform local averages with a certain group (or \textit{quorum}) size $S$,
leveraging the fact that several averaging steps can be performed in parallel. 
Removing the global communication bottleneck in decentralized SGD has been shown to enable scaling to tens and even hundreds of nodes~\cite{syncring, asyncring, assran2018stochastic}. 
However, performing averaging in pairs does come at the cost of \emph{worse convergence}: 
in particular, early proposals on decentralized algorithms~\cite{syncring, asyncring} lose accuracy with respect to the synchronous baseline at scale, 
while more recent work~\cite{assran2018stochastic, nadiradze2019popsgd} observe that the algorithms can achieve full accuracy if executed for \emph{more iterations} than the synchronous baseline: in particular, they execute between twice and four times more SGD iterations in total, relative to the synchronous baseline, erasing much of the speedup due to increased scalability. 
This decreased convergence behavior is connected to the analytical bounds provided by these algorithms: while the theoretical convergence rates suggest linear speedup with respect to the number of SGD steps and executing nodes, 
these rates only apply after a very large number of SGD steps have been taken, in order to allow the pairwise averaging process to ``mix'' well, thereby simulating all-to-all averaging. See Section~\ref{convergeproof} for a detailed discussion.

\subsection{Training at Scale}

An orthogonal challenge to distributed stochastic optimization is that of unbalanced workloads. Imbalance may be caused by the training system~\cite{iosup2011performance,li2020taming,asyncring} or by the task itself~\cite{li2020taming,wijmans2019decentralized}. Training on multi-tenant cloud systems can suffer from performance variability due to resource sharing. Several deep learning tasks, such as video classification and machine translation, have inherent load imbalance, because input/output sequences have different lengths~\cite{li2020taming}. In deep reinforcement learning, an agent must interact with the environment to generate training data. For RL tasks using heterogeneous environments~\cite{wijmans2019decentralized}, the runtime of training data generation varies significantly. Quantitative profiling for the load imbalance of language modeling and RL is presented in Section~\ref{eval:transformer} and Section~\ref{eval:drl}, respectively.

\subsection{Comparison Targets}
\label{relatedwork}

In Table~\ref{tab:rel-sgd} we summarize and classify the distributed SGD algorithms most relevant to our work. Algorithms in \textbf{bold} are used for comparison in this work. Since decentralized algorithms typically scale and perform better on large-scale systems than centralized algorithms, we limit our comparison to decentralized algorithms. 
The algorithms with which we compare our approach are chosen specifically to be spread across the different answers to the above five questions, prioritizing popular algorithms with proven convergence, both in theory and in practice:

\begin{itemize}
  \item Allreduce-SGD is the standard data-parallel training.
  \item Local SGD~\cite{stich2018local,lin2018don,wang2019adaptive} performs a fixed number of local iterations of SGD (a hyperparameter determined by the user) and then averages the models over all processes with a standard allreduce. Several variants with different methods for determining the frequency of global averaging exist.
  \item Decentralized parallel SGD (D-PSGD)~\cite{syncring} uses a ring topology, where each process averages its local model with its two neighbors. Processes advance synchronously with a single global clock.
  \item Stochastic gradient push (SGP)~\cite{assran2018stochastic} generalizes the topology used in D-PSGD to support more flexible, asymmetric communication patterns.
  \item Eager-SGD~\cite{li2020taming} uses partial collective allreduces over the gradients, allowing at most half processors to contribute stale gradients if not ready.
  \item For Asynchronous decentralized parallel SGD (AD-PSGD)~\cite{asyncring}, each worker runs a communication thread and a computation thread in parallel. The computation thread calculates the model updates while the communication thread exchanges models with a selected communication neighbor. AD-PSGD extends the idea of D-PSGD by allowing workers to exchange models at any time.  
\end{itemize}

These cover nearly all varieties of consistency and averaging, and practical differences in communication patterns. We do not include a comparison to a decentralized, asynchronous, gradient averaging algorithm, as GossipGraD does not show good convergence for many applications and Choco-SGD has not yet been applied to deep learning.

\subsection{Discussion}

Following the discussion on the impact of quorum size on convergence (\textbf{Q5}), it is natural to ask whether performing decentralized averaging \emph{in larger groups} would be able to provide the best of both worlds, enabling the full convergence of the synchronous algorithm, and the scalability of fully decentralized ones. 
There are two main barriers to this solution: the first one is at the implementation level, since, to our knowledge, no efficient non-blocking implementation of group model averaging exists. 
The second is at the application level, since it is not clear whether group averaging would be able to achieve the same convergence as the synchronous solution (both in theory and in practice). 
In the following sections, we address both of these issues.

\section{Wait-Avoiding Group Communication}

The allreduce operation~\cite{mpi-3.1} is defined as a reduction whose results are shared among all participants. 
Although several
optimizations~\cite{gibiansky2017bringing,li2013numa,rabenseifner2004optimization,li2014improved}
have been designed to improve the performance of this collective, allreduce
poses an implicit global synchronization point, which makes the operation vulnerable
to stragglers during deep learning training. 
On larger systems, the performance of the compute nodes can be impacted by
different internal (e.g., load imbalance) and external factors (e.g., OS or network~\cite{de2019mitigating} noise),
potentially increasing the synchronization overhead.
We define this collective as \textit{synchronous allreduce}. While
non-blocking collectives~\cite{hoefler2007implementation} can alleviate the
synchronization overhead, they do not fully remove it and completion still waits.
Even if the participating processes are perfectly synchronized, the optimal
scaling of an allreduce of size $N$ is at best $\mathcal{O}\left(\log P +
N\right)$ for $P$ processes~\cite{allred-bounds-yuan,sparcml}. Therefore,
growing process counts will reduce the parallel efficiency and eventually make
the reduction a scaling bottleneck.

\subsection{Wait-Avoiding Group Allreduce}
\label{waitfree}

To overcome the synchronization overhead and
overall collective cost, we introduce a new class of \emph{wait-avoiding group
collectives}, focusing on \emph{group allreduce} for the purpose of this work.  We relax
synchronization by making the collectives 
externally-triggerable~\cite{di2015exploiting, li2020taming}, namely, a collective
can be initiated without requiring that all processes enter it, by
externally activating the communication schedule of \emph{late} processes with
activation messages sent by the \emph{early} ones. 
Once activated, a
group allreduce does not perform a global reduction. Instead, it
partially reduces the data within non-overlapping groups of processes, limiting
the number of communications needed to implement the group collective.

\subsubsection{Collective activation}
In a wait-avoiding group allreduce, any process can make progress regardless of
what the other processes are working on. This wait-avoidance is achieved by
the activation component. We call the process reaching the collective call
first the \textit{activator}. The activator is in charge of informing the
other processes that an allreduce operation has started and that they have to participate, regardless of whether they reached the collective call-site.

\begin{figure}[ht!]
  \centerline{\includegraphics[width=.9\linewidth]{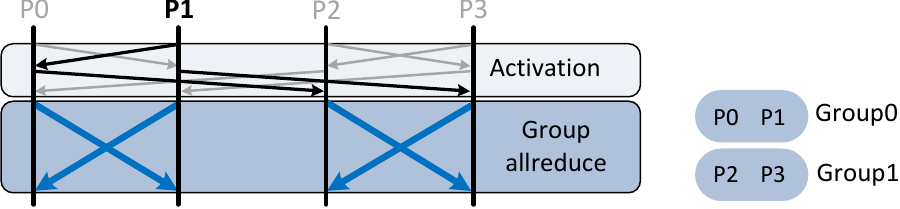}}
  \caption{Wait-avoiding group allreduce on four processes with a group size of two. P1 arrived first and activates the operation.}
  \label{waitfreegroup}
\end{figure}

In a wait-avoiding group allreduce, any process can initiate the collective.  We use a modified version
of the recursive doubling algorithm that builds a butterfly topology, which can
be seen as a set of overlapping binomial trees, one rooted at each process.
Any node can activate the collective by sending activation messages along
the binomial tree rooted at itself. 
Fig.~\ref{waitfreegroup} shows an example where P1 is the activator. In this
case, P1 uses its broadcast tree and sends the activation messages to
P0 and P3. Once activated, P0 first forwards the activation message to P2,
after which P0 starts executing its group allreduce schedule. 

It is possible that several processes arrive at the wait-avoiding group allreduce
operation at close proximity, which means we may have more than one activator during the activation phase. To
guarantee that a process does not execute the same collective twice, we assign
each operation a version number that is increased every time the collective
is executed. The collective version number is encoded in the tag of the
activation messages: once an activation is received, the collective schedule is
activated only if its version number is lower than or equal to the one carried by
the activation message. The version number check is executed also when a
process reaches the collective: if it fails, then the version of the
collective that the process wants to activate has already been executed (and the
process has passively participated in it). In this case, no activation messages
are sent.

\subsubsection{Asynchronous execution}
To enable asynchronous execution of the custom collectives, we extend the \texttt{fflib} communication
library~\cite{di2015exploiting}, adding support for wait-avoiding group allreduce.
\texttt{fflib} allows programmers to customize collective operations via a flexible, DAG-based representation of point-to-point and local compute operations, defined as \emph{schedules}. The library provides a C-based interface for schedule creation and nonblocking invocation, using MPI as its primary backend, with additional support for network offloading engines such as sPIN~\cite{spin}. Our defined schedule for group operations models both the activation and group allreduce phases.

\subsection{Dynamic Grouping Strategy}
\label{groupingstrategy}

As discussed in Section \ref{sec:bg}, in data-parallel SGD variants such as allreduce SGD~\cite{sergeev2018horovod,deep500} and gossip SGD~\cite{syncring,assran2018stochastic,asyncring}, each process keeps propagating local model updates to the other processes at every iteration to make global progress. We propose a dynamic grouping strategy to reduce the latency (in steps) of local update propagation. Together with 
the group allreduce operation, the grouping strategy guarantees that the local updates can be globally propagated 
within $\log P$ iterations. The larger the group size, the faster the updates are propagated. By carefully selecting the group size, we can achieve both lower latency than gossip SGD and efficient communication by reducing contention.

We define the dynamic grouping strategy in Algorithm~\ref{alg:grouping}.
We assume the number of processes $P$ is a power-of-two, which is a common case in current distributed training systems. The group size $S$ ($\le P$) is also set to a power-of-two. 
In line 2, we initialize the \textit{mask}, and calculate the number of phases in a butterfly topology for $P$ and $S$ processes, respectively.
Line 3 initializes the \textit{shift}. In each training iteration $t$, the algorithm first initializes $P$ groups, each of which contains one process (line 4). In line 8, an equivalence relation between each pair of processes is found using the bitwise XOR operation. For a pair of processes with an equivalence relation (i.e., $p\equiv q$), we find the groups $p$ and $q$ belong to, respectively (line 9); if $p$ and $q$ are not
in the same group, we merge the two groups into one using the union operation (lines 10--12). In line 15, the processes will have been partitioned into $P/S$ groups in iteration $t$. Note that the initial value of \textit{shift}
is periodically changing in every iteration (line 3), which, in turn, changes the group composition in every iteration. 

\begin{algorithm}[h!]
  \footnotesize
  \begin{algorithmic}[1]
    \State \textbf{Input:} Total $P$ processes. $S$ is the group size. $t$ is the training iteration.
    \State \textit{mask} = 1,\ \  \textit{global\_phases} = $log_{2}P$,\ \  \textit{group\_phases} = $log_{2}S$
    \State \textit{shift} = $(t*$\textit{group\_phases}$)$$\mod$\textit{global\_phases}
    \State \textbf{\textit{make\ each\ process\ an\ individual\ group}} \Comment{initialize $P$ groups}
    \For{$r = 1$ \textbf{to} \textit{group\_phases}}
    \State \textit{mask} $<<=$ \textit{shift} \Comment{bitwise left shift on \textit{mask}}
    \For{$p = 0$ \textbf{to} $P-1$}
    \State $q = p$ XOR \textit{mask} \Comment{equivalence relation $p\equiv q$}
    \State \textbf{Find groups}: $p\in$ \textit{group\_p}, $q \in$ \textit{group\_q}
    \If{\textit{group\_p} $\not=$ \textit{group\_q}}
    \State \textbf{Merge groups}: \textit{group\_merge} $=$ \textit{group\_p}$\ \cup\ $\textit{group\_q}
    \EndIf
    \EndFor
    \State \textit{shift} $= ($\textit{shift}$+1)\mod$\textit{global\_phases}
    \EndFor \Comment{processes are partitioned into $P/S$ groups in iteration $t$}    
  \end{algorithmic}
  \caption{Dynamic grouping strategy}
  \label{alg:grouping}
\end{algorithm}

To demonstrate dynamic grouping, we use $P=8$ and $S=4$ as an example. In iteration $0$,
all processes are initially partitioned into 8 groups. The set of equivalence relations
includes $0\equiv1$, $2\equiv3$, $4\equiv5$, $6\equiv7$, $0\equiv2$, $1\equiv3$, $4\equiv6$, and 
$5\equiv7$. By recursively merging the two groups in which a pair of processes with a equivalence relation belongs to, we obtain two non-overlapping groups, which contain the processor sets \{0, 1, 2, 3\} and \{4, 5, 6, 7\}. 
In iteration $1$, the set of equivalence relations changes; thus, the grouping changes accordingly (i.e., \{0, 1, 4, 5\} and \{2, 3, 6, 7\}).

Note that we only use Algorithm~\ref{alg:grouping} to formally describe the grouping strategy.
The grouping strategy together with allreduce within each group is implemented concisely following the phases
of the butterfly topology, namely each pair of processes with a equivalence relation in a phase would exchange messages. We use the variable $t$ to change the phases that should be executed in the current iteration.
Fig.~\ref{iterative} presents the iterative execution of group allreduce with
dynamic grouping in WAGMA-SGD, and grouping is shown on the right side. We can see that although the group size is fixed, the groups are dynamically changing during the iterations. 
Within each group, the allreduce is conducted following $\log_2 S$ phases of the butterfly topology. To maintain convergence with this communication scheme in data-parallel deep learning training, a standard synchronous allreduce across all processes is conducted every $\tau$ iterations, bounding the staleness of the weights. 
In the following section, we will present the algorithm in detail and discuss this periodic synchronization further.

\begin{figure}[t]
\centerline{\includegraphics[width=.9\linewidth]{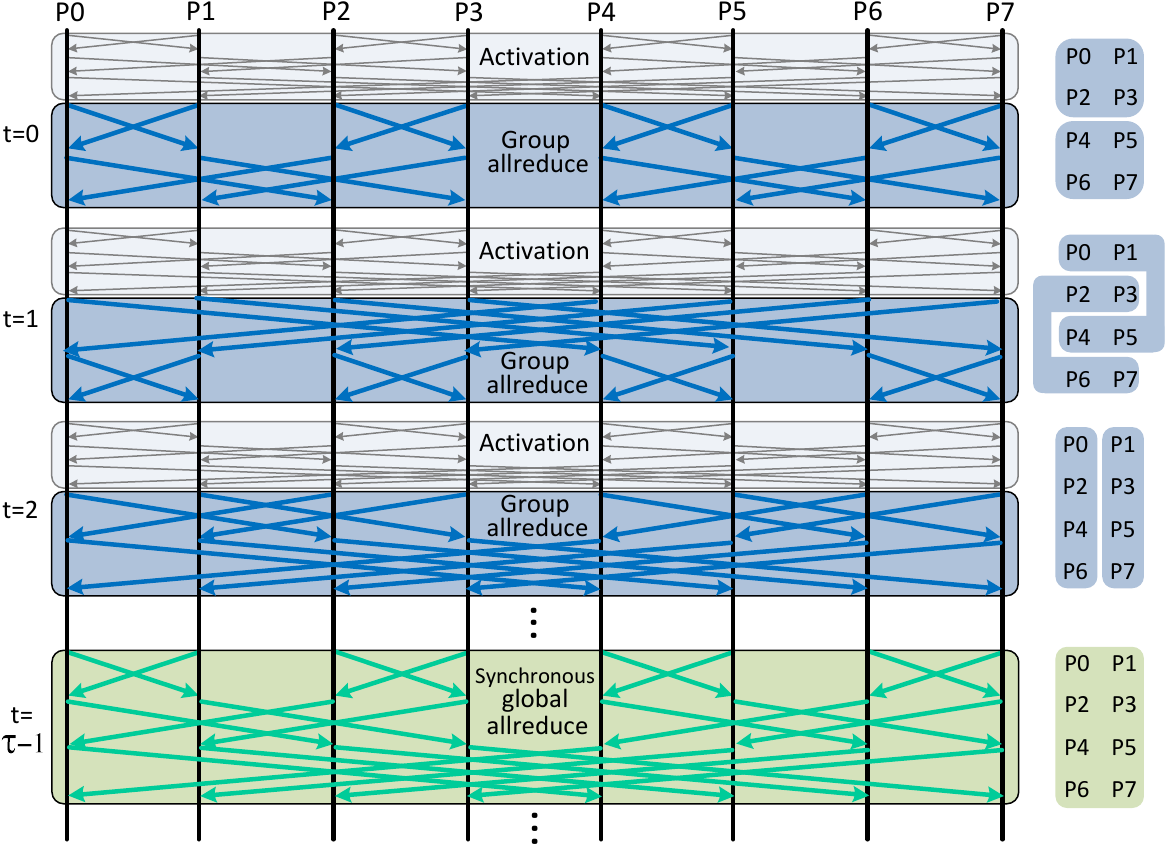}}
\caption{Communication scheme of WAGMA-SGD. Total 8 processes and the group size is 4. Every $\tau$ iterations, the algorithm synchronizes globally.}
\label{iterative}
\end{figure}

\section{Wait-Avoiding Group Model Averaging}

\begin{algorithm}[t]
  \footnotesize
  \begin{algorithmic}[1]
    \State \textbf{Input:} $b$ is local batchsize for $P$ processes. $S$ is the group size of the processes. 
    $\tau$ is the synchronization period. 
    \For{$t = 0$ \textbf{to} $T-1$}
    \State \textcolor{Orange}{$\vec{x},\vec{y}\leftarrow$ Each process samples $b$ elements from dataset}
    \State \textcolor{Orange}{$\vec{z} \leftarrow \ell\left(W_{t}, \vec{x}, \vec{y}\right)$}
    \State \textcolor{Orange}{$G_t^{local}\leftarrow\frac{1}{b} \Sigma_{i=0}^{b}\nabla\ell\left(W_{t}, \vec{z}_i\right)$}
    \State \textcolor{Orange}{$\Delta W_{t}\leftarrow U\left(G_t^{local},W_{(0,\dots,t)}, t\right)$}
    \State \textcolor{Orange}{$W_{t}^{\prime} \leftarrow W_{t}+\Delta W_{t}$}
    \If{$(t+1)\mod\tau \not= 0$}
    \State \textcolor{RoyalBlue}{$W_{t}^{sum}\leftarrow\ wait$-$avoiding\_group\_allreduce\left(W_{t}^{\prime}, t\right)$}
        \If{$W_{t}^{\prime}\ is\ not\ stale$}
        \State $W_{t+1}\leftarrow\frac{1}{S}\ W_{t}^{sum}$
        \Else
        \State $W_{t+1}\leftarrow\frac{1}{S+1}\ \left(W_{t}^{sum} + W_{t}^{\prime}\right)$
        \EndIf
    \Else
    \State \textcolor{JungleGreen}{$W_{t+1}\leftarrow\frac{1}{P}\ sync\_allreduce\left(W_{t}^{\prime}\right)$}
    \EndIf    
    \EndFor
  \end{algorithmic}
  \caption{WAGMA-SGD}
  \label{alg:WAGMAsgd}
\end{algorithm}

Based on the insight that larger groups converge faster, and on the novel implementation of wait-avoiding group collectives, we design the Wait-Avoiding Group Model Averaging (WAGMA) SGD algorithm. WAGMA-SGD can be classified as a \textit{model-averaging}, \textit{bounded-staleness}, \textit{decentralized} SGD with a \textit{group size} of $S\propto \sqrt{P}$ and a \textit{global communication period} of $\tau$ steps. As listed in Algorithm~\ref{alg:WAGMAsgd}, WAGMA-SGD is similar to minibatch SGD, but makes a few important distinctions. In lines 3--7, each process calculates the local gradients $G_t^{local}$ and then applies the local gradients to derive and apply the model update $\Delta W_t$. Subsequently, the wait-avoiding group model averaging
is conducted (lines 8--17) using the aforementioned wait-avoiding communication scheme. 
From an algorithmic perspective, WAGMA-SGD does not rely on certain choice of group members for the local collectives. However, instead of randomly choosing groups of processes, we use the 
butterfly strategy (Algorithm~\ref{alg:grouping}) for efficient and deterministic communication.

In each iteration, faster processes will trigger the model averaging immediately without waiting (line 9, $t$ is used to control grouping), which may incur averaging the local models with some stale models from the slower processes. To both bound the staleness and mitigate divergence across local model replicas, we define a synchronization period $\tau$, in which the models are averaged 
across all processes using a global allreduce (line 16). Empirically, we set the 
synchronization period $\tau$ to 10 training iterations, which balances model accuracy with training throughput, as we will show in Section~\ref{evaluate}.

From the algorithm perspective, there are four parameters in WAGMA-SGD, which can be used to approximately represent other parallel SGD variants by adjusting the parameters. The parameters include $\tau$, a boolean variable $\alpha$ to indicate whether to use the wait-avoiding group allreduce, a boolean variable $\beta$ to indicate whether to use a normal group allreduce (without activation component), and the group size $S$. Note that $\alpha$ and $\beta$ can not be \textit{true} at the same time. If $\alpha$ == $\beta$ == \textit{false}, WAGMA-SGD degrades to local SGD with global communication period $\tau$. If $\alpha$ == \textit{true}, $\beta$ == \textit{false}, S == O(1) and $\tau$ == $\infty$, WAGMA-SGD is similar to AD-PSGD but WAGMA-SGD requires to wait for the fastest process (i.e., the \textit{activator}). If $\alpha$ == \textit{false}, $\beta$ == \textit{true}, S == O(1) and $\tau$ == $\infty$, WAGMA-SGD is similar to D-PSGD and SGP.

An execution snapshot of WAGMA-SGD ($P=4$ and $S=2$) is presented in Fig.~\ref{snapshot}. 
Suppose P1 is a straggler. When the group allreduce in iteration $t$ is triggered by any of the other
three processes, P1 contributes the stale model parameters $W^{1^\prime}_{t-1}$ while the other three processes contribute up-to-date models. In iteration $t$, P1 and P0 are in
the same group; therefore, $W^{1^\prime}_{t-1}$ and $W^{0^\prime}_{t}$ will be added together to derive $W^{\text{sum}}_{t}$. P0 will
use the averaged model $W^{0}_{t+1}$ (line 11 in Algorithm~\ref{alg:grouping}) for the next iteration of training. P1 subsequently finishes the calculation for the local updated
model in iteration $t$ (i.e., $W^{1^\prime}_{t}$), but finds out that the group allreduce in iteration $t$ is already finished by checking the value of a preset tag variable in the receive buffer. In
this case, P1 will average the stale model $W^{1^\prime}_{t}$ with $W^{sum}_{t}$ (line 13 in Algorithm~\ref{alg:grouping}), and the averaged model $W^{1}_{t+1}$ will be used for the next iteration of training. Meanwhile, the data in the send buffer of P1 is updated by $W^{1^\prime}_{t}$. If the group allreduce in iteration $t+1$ is triggered by some faster process at this time, P1 will continue to passively contribute the stale model $W^{1^\prime}_{t}$. When a standard allreduce is called
at the synchronization point, all processes are forced to contribute the model parameters after training for the same number of iterations. In Fig.~\ref{snapshot}, P1 catches up with the other processes in iteration $t+1$; thus, P1 will contribute the timely model $W^{1^\prime}_{t+1}$ to P3, as they are in the same group.

\begin{figure}[t]
\centerline{\includegraphics[width=.9\linewidth]{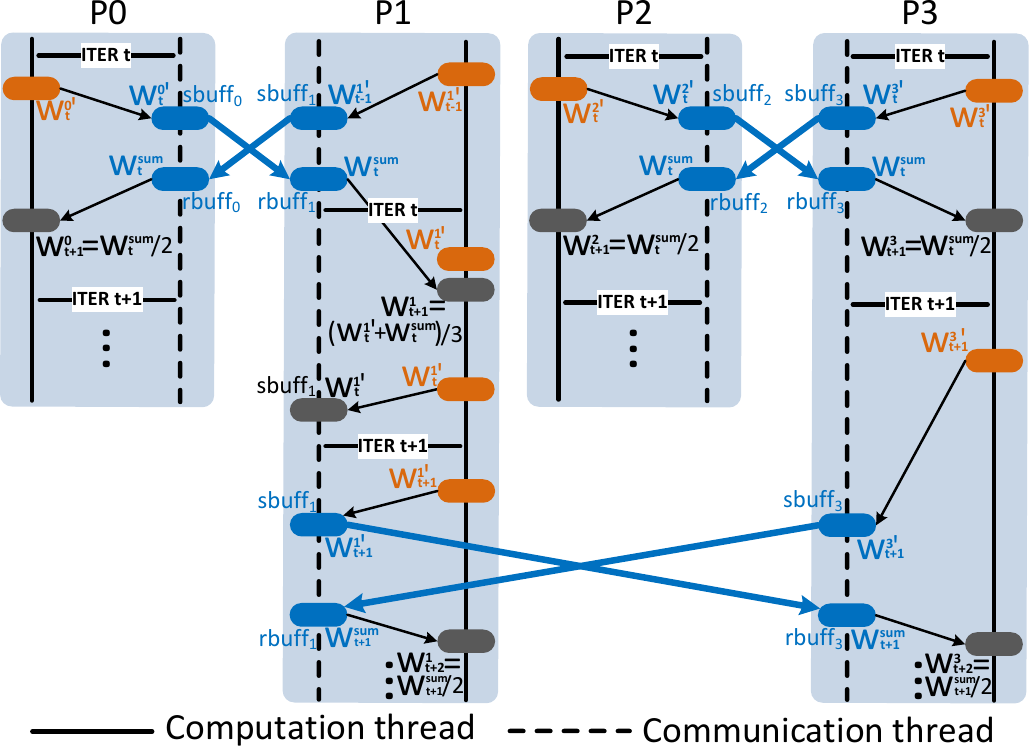}}
\caption{WAGMA-SGD execution snapshot for $P$=4 and $S$=2.}
\label{snapshot}
\end{figure}

\subsection{Proof of Convergence}
\label{convergeproof}

\subsubsection{Algorithm Modelling}
For analysis purposes, we will model the algorithm's execution as follows. 
We will proceed in steps, indexed by time $t \geq 0$. 
Each node $i$ maintains its own local model $W_t^i$, and has a local partition of the data. 
In each step, a group of nodes of size $S$ is chosen to interact. Each node takes a local gradient step, and then nodes average their models. 
This averaging step might be inconsistent, as per the above semantics.
In the analysis, we will assume that the group of $S$ interacting nodes is chosen uniformly at random---in the long run, the resulting interaction graph will have the same properties as the butterfly interaction strategy used in the implementation. 
While our analysis considers each interaction sequentially, in practice $\Theta( P / S ) $ interaction steps can occur in parallel. 

\subsubsection{Setup and Analytic Assumptions}
We will assume a standard setting in which we are given a dataset of $D$ samples $\mathcal{D} = \{e_1, e_2, \ldots, e_D\}$, and to each associate a differentiable loss function  $f_e: \mathbb{R}^d  \rightarrow \mathbb{R}$. Each node $i$ is given a random partition $\mathcal{D}_i$ of the dataset $\mathcal{E}$, and we wish to solve an empirical risk minimization problem by finding 
\[
  x^\star = \textnormal{argmin}_{x \in \mathbb{R}^d} \left[ F (x) := \frac{1}{D} \sum_{e \in \mathcal{D}} f_e(x) \right].
\]

Let $F_i = \frac{P}{D} \sum_{e \in \mathcal{D}_i} f_e $ be the loss function corresponding to the dataset of the $i$th node, and $F^\star = F(x^\star)$. 
To make the analysis tractable, we will make the following standard assumptions on the loss function.
\begin{assumption}
  \label{lem:assumption}
  We assume the following hold:
  \begin{enumerate}
    \item \emph{(Lipschitz gradients)} All functions $f_e$ have $L$-Lipschitz gradient, for some constant $L$. 
    \item \emph{(Bounded Second Moment)} There exists a constant $M$ such that  for any node $i$ and $w \in \mathbb{R}^d,$ $\mathbb{E}_{e \in D_i} \| \nabla f_e (w)  \|^2 \leq M^2$. 
    \item \emph{(Bounded Staleness)} The staleness during the averaging step is upper bounded by a parameter $\tau$. 
    That is, for any node $i$ participating in the averaging step at time $t$, averaging is performed with respect to model $W_{t'}^i$, where $t' \geq t - \tau + 1$, and every gradient update is applied globally at most $\tau$ steps after it was generated.   
    
  \end{enumerate}

\end{assumption}

\subsubsection{Convergence result}
We can now state the main convergence result.
For readability, we state a simplified variant that highlights the relationship between parameter values, in particular the relation between the convergence time $T$, $P$ processors, and the size of the interacting group $S$. The detailed proof will be presented in Section~\ref{detailedproof}.

\begin{theorem}
  \label{thm:convergence}
  Consider the setting described above, in which we wish to optimize a non-convex function $F : \mathbb{R}^d \rightarrow \mathbb{R}$. 
  Let $S$ be the size of a communication group, and assume that the maximum staleness $\tau$ is constant. 
  Fix a success parameter $\varepsilon > 0$. 
  For each time $t$, we define $\mu_t = \sum_{i = 1}^P W_t^i$ to be the average of local models at time $t$. 
  Then there exists a setting of the learning rate in the order of $P / \sqrt T$ such that, if the algorithm has taken
  
  $$T = \Omega \left( \max \left\{ \frac{(F(W_0) - F^\star)^2 }{\epsilon^2},   \frac{S^4 M^4 L^4}{\epsilon^2}, {P^4 \tau^4} \right\} \right) \textnormal{ steps, }$$ 
  
  \noindent then there exists an iterate $0 \leq T^\star \leq T$ such that
  
  $$ \mathbb{E} \| \nabla F( \mu_{T^\star} ) \|^2 \leq \varepsilon,$$
  
  \noindent where the expectation is taken w.r.t. the randomness in the sampling and interactions.

\end{theorem}

At a high level, this claim shows that the algorithm will eventually reach a point where the model average has negligible gradient, i.e., is at a local minimum. 
While this does not guarantee convergence to a global minimum, it matches the best achievable guarantees for SGD in the non-convex setting~\cite{ghadimi2013stochastic}. 
The convergence proof, provided as additional material, generalizes the  decentralized asynchronous framework of Nadiradze et al.~\cite{nadiradze2019popsgd}  since we allow for group communication, whereas they only consider pairwise interactions.

It is interesting to examine the rate at which convergence occurs --- for standard parameter settings, i.e. constant $\tau$, $S, M$ and $L$, 
the convergence (i.e., the rate at which we get to a point of negligible gradient) matches that of standard SGD, and the speedup  with respect to the number of nodes can be  \emph{linear} in $P$. 
More precisely, we can examine the three terms in the lower bound on the number of steps $T$ to convergence given by the Theorem. 
The first is the standard convergence rate for SGD. (Recall however that we express this bound in terms of \emph{total} steps, whereas $P / S$ such steps can occur in parallel.)
The second bounds the impact of the variance on the total number of iterations, whereas the third suggests that, to negate the impact of asynchrony, the algorithm has to execute for at least  $P^4 \tau^4$ steps. 
The convergence ensured by the first term is the best possible, and matches the rates for other decentralized algorithms, e.g.~\cite{syncring, asyncring}.

\subsubsection{Proof}
\label{detailedproof}

In this section, we focus on providing a formal proof for the convergence statement provided in Theorem~\ref{thm:convergence}.

\macsection{Preliminaries:} Let $\mathcal{K}$ be a set all subsets of  $\{1,2,...,P\}$, which have size $S$.
if set $K \in \mathcal{K}$ is chosen for interaction. we have that for each $i \in K$
\begin{align*}
W_{t+1/2}^i=W_t^i/S+\sum_{j \in K/\{i\}} {W_t^j}'/S.
\end{align*}
and
\begin{align*}
W_{t+1}^i=W_{t+1/2}^i-\eta \tilde{G}(W_t^i) \\
\end{align*}
Where for each process $i$, if $t \ge q_t^i > t-\tau$ is the last time it interacted before the step $t$, we have 
\begin{equation}
{W_t^i}'=W_{q_t^i+1/2}^i=W_{q_t^i+1}+\eta \tilde{G}(W_{q_t^i}^i)=W_t^i+\eta\tilde{G}(W_{q_t^i}^i).
\end{equation}

Let $\E_t$ be expectation which is conditioned on the entire history up to and including step $t$.
This means that 
\begin{align} \label{eqn:mulinear}
\E_t[\mu_{t+1}]&=\mu_t+ \nonumber\\&\frac{1}{\binom{P}{S}} \sum_{K \in \mathcal{K}} \sum_{i \in K} \E_t \Big[-\frac{\eta\tilde{G}(W_t^i)}{P}+\frac{\eta(S-1)\tilde{G}(W_{q_t^i}^i)}{PS}\Big] \nonumber\\&=\mu_t+\frac{S}{P} \sum_{i=1}^P
\Big(-\frac{\eta\nabla f(W_t^i)}{P}+\frac{\eta(S-1)\tilde{G}(W_{q_t^i}^i)}{PS}\Big).
\end{align}

\macsection{Potential Bound:} Now we derive bound on the potential $\Gamma_t=\sum_{i=1}^P \|W_t^i-\mu_t\|^2$ in expectation.
Observe that for each $t$, $\Gamma_{r_t}=0$, where $t \ge r_t > t-\tau$
is the last time global averaging happened.
Hence we need to bound the change over $\tau$ steps.
Using induction we show that:
\begin{lemma} \label{lem:fixedi}
for each $t \ge 0$ and process $i$:
\begin{equation}
\E\|W_t^i-\mu_{r_t}\|^2 \le 4\eta^2M^2(t-r_t)^2 \le 4\eta^2\tau^2M^2.
\end{equation}
\end{lemma}
\begin{proof}
Base case $t=r_t$ holds trivially. We assume that the claim holds for $t$ (and every process $i$), and we show that it holds for $t+1$.
Let $K$ be a set chosen for averaging. 
If $i \in K$ then  we have that
\begin{align*}
W_{t+1}^i=\sum_{j \in K} W_t^j/S-\eta \tilde G(W_t^i)+\sum_{j \in K/\{i\}} \eta \tilde G(W_{q_t^j}^j)/S.
\end{align*}
Hence, if $t > r_t$
\begin{align*}
\|W_{t+1}^i&-\mu_{r_t}\|^2=\\&\Big\|\sum_{j \in K} (W_t^j-\mu_{r_t})/S-\eta \tilde G(W_t^i)+\sum_{j \in K/\{i\}} \eta \tilde G(W_{q_t^j}^j)/S \Big\|^2 \\&\overset{Young}{\le} \frac{1}{S^2} (1+\frac{1}{t-r_t})\Big\|\sum_{j \in K} (W_t^j-\mu_{r_t}) \Big\|^2
+\\&\quad\quad(1+t-r_t)\Big\|-\eta \tilde G(W_t^i)+\sum_{j \in K/\{i\}} \eta \tilde G(W_{q_t^j}^j)/S\Big\|^2
\\ &\overset{Cauchy-Schwarz}{\le} \frac{1}{S}(1+\frac{1}{t-r_t})\sum_{j \in K} \Big\|W_t^j-\mu_{r_t}\Big\|^2
\\&\quad\quad\quad+2(1+t-r_t)\eta^2\Big\| \tilde G(W_t^i)\Big\|^2\\&\quad\quad\quad+\frac{2\eta^2(1+t-r_t)}{S^2}\Big\|\sum_{j \in K/\{i\}} \tilde G(W_{q_t^j}^j)\Big\|^2
\\ &\overset{Cauchy-Schwarz}{\le} \frac{1}{S}(1+\frac{1}{t-r_t})\sum_{j \in K} \Big\|W_t^j-\mu_{r_t}\Big\|^2
\\&\quad\quad\quad+2(1+t-r_t)\eta^2\Big\| \tilde G(W_t^i)\Big\|^2\\&\quad\quad\quad+\frac{2\eta^2(1+t-r_t)(S-1)}{S^2}\sum_{j \in K/\{i\}} \Big\|\tilde G(W_{q_t^j}^j)\Big\|^2
\end{align*}
Next we take expectations:
\begin{align*}
\E\|W_{t+1}^i&-\mu_{r_t}\|^2 \le \frac{1}{S}(1+\frac{1}{t-r_t})\sum_{j \in K} \E\Big\|W_t^j-\mu_{r_t}\Big\|^2
\\&\quad+2(1+t-r_t)\eta^2\E\Big\| \tilde G(W_t^i)\Big\|^2\\&\quad+\frac{2\eta^2(1+t-r_t)(S-1)}{S^2}\sum_{j \in K/\{i\}} \E\Big\|\tilde G(W_{q_t^j}^j)\Big\|^2
\end{align*}
Using the second moment bound we get
\begin{align*}
\E\|W_{t+1}^i&-\mu_{r_t}\|^2 \le \frac{1}{S}(1+\frac{1}{t-r_t})\sum_{j \in K} \E\Big\|W_t^j-\mu_{r_t}\Big\|^2
\\&\quad\quad\quad\quad\quad+2(1+t-r_t)\eta^2M^2\\&\quad\quad\quad\quad\quad+\frac{2\eta^2(1+t-r_t)(S-1)^2M^2}{S^2} \\
&\quad\quad\quad\quad \le \frac{1}{S}(1+\frac{1}{t-r_t})\sum_{j \in K} \E\Big\|W_t^j-\mu_{r_t}\Big\|^2\\&\quad\quad\quad\quad\quad+4(1+t-r_t)\eta^2M^2
\end{align*}
Finally using the assumption that claim holds for $t$ (and every $i$), we get that 
\begin{align*}
 \E\|W_{t+1}^i&-\mu_{r_t}\|^2 \le 4\eta^2M^2(t-r_t)^2(1+\frac{1}{t-r_t})\\&\quad\quad\quad\quad\quad+4(1+t-r_t)\eta^2M^2
 \\ &\quad\quad\quad\quad= 4(t-r_t+1)^2\eta^2M^2.
\end{align*}

If $i$ does not belong to $K$, claim for $t+1$ holds trivially (since $i$ does not perform SGD step).

\end{proof}
This allows us to show the following lemma:
\begin{lemma} \label{lem:fixedmu}
\begin{equation*}
\E\|\mu_t-\mu_{r_t}\|^2\le 4\eta^2M^2\tau^2.
\end{equation*}
\end{lemma}
\begin{proof}
Observe that since $\mu_t=\sum_{i=1}^P W_t^i/P$, by Jensen's inequality we have that 
\begin{equation*}
\E\|\mu_t-\mu_{r_t}\|^2\le \sum_{i=1}^P \E\|W_t^i-\mu_{r_t}\|^2/P \overset{\text{Lemma \ref{lem:fixedi}}}{\le}
4\eta^2M^2\tau^2.
\end{equation*}
\end{proof}

Next we need another lemma which upper bounds discrepancy between 
$\mu_t$ and $W_{q_t^i}^i$ for each process $i$:

\begin{lemma} \label{lem:muandmodeloutdated}
For each $t \ge 0$ and process $i$ we have that:
\begin{equation*}
\E\|W_{q_t^i}^i-\mu_t\|^2 \le 16\eta^2M^2\tau^2.
\end{equation*}
\end{lemma}
\begin{proof}
Recall that for each $t$ and $i$, $t \ge q_t^i > t-\tau$.
Hence we have that $r_t=r_{q_t^i}$.
From Lemmas \ref{lem:fixedmu} and \ref{lem:fixedi} it follows that
$\E\|\mu_{r_t}-W_{q_t^i}^i\|^2 \le 4\eta^2M^2\tau^2$ and $\E\|\mu_t-\mu_{r_t}\|^2 \le 4\eta^2M^2\tau^2$.
Thus by using Cauchy-Schwarz inequality we get that 
\begin{equation*}
\begin{split}
\E\|W_{q_t^i}^i-\mu_t\|^2 \le 2\E\|W_{q_t^i}^i-\mu_{r_t}\|^2+2\E\|\mu_t-\mu_{r_t}\|^2 \\ \le 16\eta^2M^2\tau^2.
\end{split}
\end{equation*}

\end{proof}

Finally we upper bound the Gamma potential in expectation:
\begin{lemma} \label{lem:GammaBound}
For any $t \ge 0$, we have that 
\begin{equation}
\E[\Gamma_t]=\sum_{i=1}^P \E\|W_t^i-\mu_t\|^2 \le 16P\eta^2M^2\tau^2.
\end{equation}
\end{lemma}
\begin{proof}
For each process $i$, we have that:
\begin{equation}
\E\|W_t^i-\mu_t\|^2 \le 2\E\|W_t^i-\mu_r^t\|^2+2\E\|\mu_t-\mu_r^t\|^2 \le 16\eta^2M^2\tau^2.
\end{equation}
Where first inequality is Cauchy-Schwarz and second inequality is combination
of Lemmas \ref{lem:fixedi} and \ref{lem:fixedmu}.
The upper bound we need to prove the lemma follows trivially.
\end{proof}

\macsection{Convergence Proof:} We are now ready to state and prove our main theorem. We provide a slightly more general version, which implies Theorem~\ref{thm:convergence}.

\begin{theorem}
 For a smooth non-convex objective function $f$, whose minimum $x^*$ we are trying to find using WAGMA algorithm and the constant learning rate $\alpha = \frac{P}{\sqrt{T}}$, where $T \ge P^4\tau^4$ is the number of iterations:
 \begin{align*}
\frac{1}{T} \sum_{t=0}^{T-1} \E\|\nabla f(\mu_t)\|^2 &\le \frac{2(f(\mu_0)-f(x^*))}{\sqrt{T}} +\\&  \frac{8S^2M^2}{\sqrt{T}}
 +\frac{64 S^2L^2M^2}{\sqrt{T}}.
\end{align*}
 \end{theorem}

\begin{proof}
Recall that $\E_t$ is expectation which is conditioned on the entire history up to and including step $t$.
By descent Lemma we know that
\begin{equation} \label{eqn:descentsmoothness}
\begin{split}
\E_t[f(\mu_{t+1})] \le f(\mu_t)+\E_t\langle\nabla f(\mu_t), \mu_{t+1}-\mu_t \rangle \\ +\frac{L}{2} \E_t\|\mu_{t+1}-\mu_t\|^2. \quad\ \ 
\end{split}
\end{equation}
First, we look at $\E_t\|\mu_{t+1}-\mu_t\|^2$:
\begin{align*}
\E_t\|\mu_{t+1}-\mu_t\|^2=\quad\quad\quad\quad\quad\quad\quad\quad\quad\quad\quad\quad\quad\quad\quad\quad \\\frac{1}{\binom{P}{S}} \sum_{K \in \mathcal{K}}  \E_t \Bigg\|\sum_{i \in K} \Big(-\frac{\eta\tilde G(W_t^i)}{P}+\frac{\eta(S-1)\tilde{G}(W_{q_t^i}^i)}{PS}\Big)  \Bigg\|^2.
\end{align*}
Next we apply Cauchy-Schwarz inequality and get:
\begin{align} \label{eqn:squarechange}
\E_t\|\mu_{t+1}&-\mu_t\|^2\le \nonumber\\& \frac{S}{\binom{P}{S}} \sum_{K \in \mathcal{K}}\sum_{i \in K}  \E_t \Bigg\| -\frac{\eta\tilde G(W_t^i)}{P}+\frac{\eta(S-1)\tilde{G}(W_{q_t^i}^i)}{PS}  \Bigg\|^2 \nonumber \\ &= \frac{S^2\eta^2}{P^3} \sum_{i=1}^P 
\E_t \Bigg\| -\tilde G(W_t^i)+\frac{(S-1)\tilde{G}(W_{q_t^i}^i)}{S}  \Bigg\|^2 \nonumber \\ &\overset{Cauchy-Schwarz}{\le} 
\frac{2S^2\eta^2}{P^3} \sum_{i=1}^P \E_t \|-\tilde G(W_t^i) \|^2+\nonumber\\&\quad\quad\quad\quad\quad\quad\quad \frac{2(S-1)^2\eta^2}{P^3} \sum_{i=1}^P\E_t \|\tilde{G}(W_{q_t^i}^i)\|^2
\nonumber \\ &\le \frac{2S^2M^2\eta^2}{P^2}+\frac{2(S-1)^2\eta^2}{
P^3} \sum_{i=1}^P\|\tilde{G}(W_{q_t^i}^i)\|^2.
\end{align}
Now, we upper bound $\E_t\langle\nabla f(\mu_t), \mu_{t+1}-\mu_t \rangle$. Using (\ref{eqn:mulinear}), we have that 
\begin{align*}
\E_t&\langle\nabla f(\mu_t), \mu_{t+1}-\mu_t \rangle=\\& \frac{\eta S}{P^2} \sum_{i=1}^P \langle \nabla f(\mu_t), -\nabla f(W_t^i)+\frac{S-1}{S}
\tilde{G}(W_{q_t^i}^i) \rangle \\&= \frac{\eta S}{P^2} \sum_{i=1}^P\langle \nabla f(\mu_t),  \frac{S-1}{S}
\tilde{G}(W_{q_t^i}^i)\rangle+\\&\frac{\eta S}{P^2} \sum_{i=1}^P\langle \nabla f(\mu_t),  \nabla f (\mu_t)-\nabla f(W_t^i) \rangle-\frac{\eta S}{P} \|\nabla f(\mu_t)\|^2 \\& \overset{Young}{\le} \frac{\eta S}{P^2} \sum_{i=1}^P\langle \nabla f(\mu_t),  \frac{S-1}{S}
\tilde{G}(W_{q_t^i}^i)\rangle+\frac{\eta }{4P}\|\nabla f(\mu_t)\|^2\\&\quad\quad+\frac{\eta S^2}{P^2} \sum_{i=1}^P\|\nabla f (\mu_t)-\nabla f(W_t^i)\|^2-\frac{\eta P}{S} \|\nabla f(\mu_t)\|^2 \\&\quad \le \frac{\eta (S-1) }{P^2} \sum_{i=1}^P\langle \nabla f(\mu_t),  
\tilde{G}(W_{q_t^i}^i)\rangle\\&\quad\quad+\frac{\eta S^2L^2}{P^2} \sum_{i=1}^P\|\mu_t-W_t^i\|^2-\frac{\eta (S-\frac{1}{4})}{P}\|\nabla f(\mu_t)\|^2.
\end{align*}
Using the above inequality and (\ref{eqn:squarechange}) in inequality (\ref{eqn:descentsmoothness}) we get that 
\begin{align*}
\E_t[f(\mu_{t+1})] &\le f(\mu_t) + \frac{2S^3M^2\eta^2}{P^3}+\frac{2(S-1)^2\eta^2}{P^3} \|\tilde{G}(W_{q_t^i}^i)\|^2\\&\ \  +\frac{\eta (S-1)}{P^2} \sum_{i=1}^P\langle \nabla f(\mu_t),  \tilde{G}(W_{q_t^i}^i)\rangle
+\\&\ \  \frac{\eta S^2L^2}{P^2} \sum_{i=1}^P\|\mu_t-W_t^i\|^2-\frac{\eta (S-\frac{1}{4})}{P}\|\nabla f(\mu_t)\|^2.
\end{align*}
Next we use $\E[f(\mu_{t+1})]=\E[\E_t[f(\mu_{t+1})]]$:
\begin{align} \label{eqn:whatever}
\E[f(\mu_{t+1})] &\le \E[f(\mu_t)] + \frac{2S^2M^2\eta^2}{P^2}\nonumber\\&
+\frac{2(S-1)^2\eta^2}{P^3} \sum_{i=1}^P \E\|\tilde{G}(W_{q_t^i}^i)\|^2\nonumber\\&+\frac{\eta (S-1)}{P^2} \sum_{i=1}^P \E\langle \nabla f(\mu_t),  
\tilde{G}(W_{q_t^i}^i)\rangle \nonumber\\& +
\frac{\eta S^2L^2}{P^2} \sum_{i=1}^P\E\|\mu_t-W_t^i\|^2 \nonumber\\& 
- \frac{\eta (S-\frac{1}{4})}{P}\E\|\nabla f(\mu_t)\|^2.
\end{align}
Recall that by the second moment bound $\E\|\tilde{G}(W_{q_t^i}^i)\|^2$ is at most $M^2$,
and by Lemma \ref{lem:GammaBound} $\sum_{i=1}^P\E\|\mu_t-W_t^i\|^2 \le 16P\eta^2M^2\tau^2$.
Hence we can rewrite (\ref{eqn:whatever}) as 
\begin{align} \label{eqn:whatever2}
\E[f(\mu_{t+1})] &\le \E[f(\mu_t)] + \frac{2S^2M^2\eta^2}{P^2}+\frac{2(S-1)^2\eta^2M^2}{P^2} \nonumber\\&+\frac{\eta (S-1)}{P^2} \sum_{i=1}^P \E\langle \nabla f(\mu_t),  
\tilde{G}(W_{q_t^i}^i)\rangle \nonumber
\\& +
\frac{16\eta^3 S^2L^2M^2\tau^2}{P} -\frac{\eta (S-\frac{1}{4})}{P}\E\|\nabla f(\mu_t)\|^2.
\end{align}

We proceed by upper bounding $\sum_{i=1}^P\E\langle \nabla f(\mu_t), 
\tilde{G}(W_{q_t^i}^i)\rangle$:
\begin{equation*}
\begin{split}
\sum_{i=1}^P\E \langle \nabla f(\mu_t),  
\tilde{G}(W_{q_t^i}^i)\rangle=\sum_{i=1}^P \E\langle \nabla f(\mu_t), \nabla f(W_{q_t^i}^i)\rangle \quad\quad\quad\quad\quad\quad
\\ =\sum_{i=1}^P \E\langle \nabla f(\mu_t),  -\nabla f(\mu_t)+
\nabla f(W_{q_t^i}^i)\rangle+P\E\|\nabla f(\mu_t)\|^2  \quad\quad
\\ \overset{Young}{\le} \frac{P}{4(S-1)}\E\|\nabla f(\mu_t)\|^2 + P\E\|\nabla f(\mu_t)\|^2 + \quad\quad\quad\quad\quad\quad
\\ (S-1)\sum_{i=1}^P \E\|\nabla f(\mu_t)-\nabla f(W_{q_t^i}^i)\|^2 \quad\quad\quad\quad\quad\quad\quad\quad
\\ \le (S-1)L^2 \sum_{i=1}^P \|W_{q_t^i}^i-\mu_t\|^2 + \frac{P(4S-3)}{4(S-1)}\E\|\nabla f(\mu_t)\|^2 \quad\quad
\\ \overset{\text{Lemma \ref{lem:muandmodeloutdated}}}{\le} 16P(S-1)L^2\eta^2M^2\tau^2 + \frac{P(4S-3)}{4(S-1)}\E\|\nabla f(\mu_t)\|^2. \quad\quad
\end{split}
\end{equation*}

Plugging the above inequality in (\ref{eqn:whatever2}) we get that 
\begin{equation*}
\begin{split}
\E[f(\mu_{t+1})] \le \E[f(\mu_t)] + \frac{2S^2M^2\eta^2}{P^2}+\frac{2(S-1)^2\eta^2M^2}{P^2}\\ +\frac{16\eta^3 (S-1)^2L^2M^2\tau^2}{P} +\frac{\eta(S-1+\frac{1}{4})}{P}\E\|\nabla f(\mu_t)\|^2
\\ + \frac{16\eta^3 S^2L^2M^2\tau^2}{P} -\frac{\eta (S-\frac{1}{4})}{P}\E\|\nabla f(\mu_t)\|^2 \quad\quad\quad\quad\ 
\\ \le \E[f(\mu_t)] + \frac{4S^2M^2\eta^2}{P^2} +\frac{32\eta^3 S^2L^2M^2\tau^2}{P} \quad\quad\quad\quad \\ -\frac{\eta}{2P}\E\|\nabla f(\mu_t)\|^2. \quad\quad\quad
\end{split}
\end{equation*}
Next we sum the above inequality for $t=0$ to $T-1$ and rearrange terms:

\begin{align*}
\frac{\eta}{2P} \sum_{t=0}^{T-1} \E\|\nabla f(\mu_t)\|^2 &\le f(\mu_0)-\E[f(\mu_T)]+\frac{4S^2M^2\eta^2T}{P^2}
 \\& +\frac{32\eta^3 S^2L^2M^2\tau^2T}{P}
\end{align*}
Now we divide the above inequality by $\frac{\eta}{2PT}$ and use the fact that $\E[f(\mu_T)] \ge f(x^*)$ :
\begin{align*}
\frac{1}{T} \sum_{t=0}^{T-1} \E\|\nabla f(\mu_t)\|^2 &\le \frac{2P(f(\mu_0)-f(x^*))}{\eta T}+\frac{8S^2M^2\eta}{P}
 \\& +64\eta^2 S^2L^2M^2\tau^2.
\end{align*}
By plugging $\eta=\frac{P}{\sqrt{T}}$ and noticing that for $T \ge P^4\tau^4$, $\eta \le \frac{1}{P\tau^2}$ we get :
\begin{align*}
\frac{1}{T} \sum_{t=0}^{T-1} \E\|\nabla f(\mu_t)\|^2 &\le \frac{2(f(\mu_0)-f(x^*))}{\sqrt{T}}  +\frac{8S^2 M^2}{\sqrt{T}} \big( 1 + 8 L^2 \big).
\end{align*}
\end{proof}

\section{Experimental Evaluation}
\label{evaluate}

We conduct our experiments on the CSCS Piz Daint supercomputer. Each Cray XC50 compute node contains a 12-core Intel Xeon E5-2690 CPU with 64 GB RAM, and one NVIDIA Tesla P100 with 16 GB memory. The compute nodes are connected by Cray Aries interconnect in a Dragonfly topology. The communication library is Cray MPICH 7.7.2. We use one MPI process per node and utilize the GPU for acceleration in all following experiments.
We evaluate three different deep learning problems, including image classification (ResNet-50~\cite{he2016deep} on ImageNet~\cite{deng2009imagenet}), machine translation (Transformer~\cite{vaswani2017attention} on WMT17), 
and deep reinforcement learning (PPO~\cite{wijmans2019decentralized} for navigation in Habitat~\cite{savva2019habitat}). For throughput tests, we run the number of nodes until reaching a point where batch size is too large to converge~\cite{shallue2018measuring}.

\subsection{Baselines}

We compare our WAGMA-SGD with the state-of-the-art data-parallel SGD variants, including Allreduce-SGD~\cite{sergeev2018horovod,deep500}, local SGD~\cite{lin2018don,stich2018local}, gossip-based SGD variants
(D-PSGD~\cite{syncring}, AD-PSGD~\cite{asyncring}, and SGP~\cite{assran2018stochastic}), and eager-SGD~\cite{li2020taming}.
Unless mentioned specifically, the synchronization period of local SGD is set to one, namely calling a standard allreduce to average the models in each training iteration, which essentially is a synchronous SGD. For SGP, we evaluate its performance with different number of communication neighbors~\cite{assran2018stochastic}. For more detailed discussion about the baselines, please refer to Section~\ref{sec:bg}.

\subsection{Image Classification with Simulated Workload Imbalance}
\label{sec:imagenet}

Residual Networks (ResNet)~\cite{he2016deep} are pervasively used in computer vision tasks. To evaluate their performance, we train ResNet-50 v1 (total 25,559,081 trainable parameters) on ImageNet using TensorFlow~\cite{tensorflow2015-whitepaper} as the basic platform. Although the training workload is balanced due to the input size being fixed, performance variability is observed when training on multi-tenant cloud systems~\cite{asyncring,li2020taming,schad2010runtime} due to resource sharing. To simulate the same degree of imbalance, we randomly select two processes at every training step to inject a certain amount of delay (320 ms). This simulated load imbalance also helps us to compare the robustness of the parallel SGD variants with the most popular deep learning benchmark. For WAGMA-SGD, we set the synchronization period $\tau = 10$, and the group size $S = \sqrt P$. 
This group setting is chosen to balance the trade-off between the faster mixing time, and improved convergence, provided by larger group size as discussed in the previous section, and the inherent increased communication overhead due to increasing $S$. In particular, this setting allows roughly $\sqrt P$ groups of $\sqrt P$ nodes each to perform reductions in parallel, on average.
Both $P$ and $S$ are power-of-two in our experimental configuration.

\begin{figure}[t]
\centerline{\includegraphics[width=.82\linewidth]{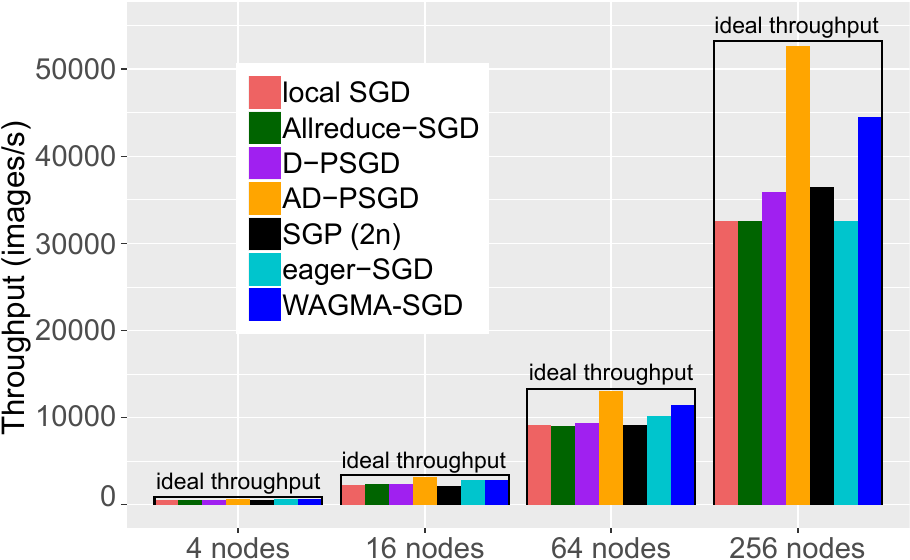}}
\caption{Throughput comparison between different parallel SGD algorithms for ResNet-50 on ImageNet with simulated load imbalance. Local batch size is 128. (AD-PSGD achieves the highest throughput but with much less accuracy.)}
\label{imagenettp}
\end{figure}

Fig.~\ref{imagenettp} shows the training throughput as the number of GPU nodes increases from 4 to 256, and the top of the rectangle
wrapping each cluster indicates the ideal throughput without communication overhead.
Compared with local SGD, Allreduce-SGD (implemented in Deep500~\cite{deep500}), D-PSGD, SGP (two communication neighbors), and eager-SGD when training on 64 GPU nodes,
WAGMA-SGD achieves 1.25x, 1.26x, 1.23x, 1.25x, and 1.13x speedup, respectively. The speedup becomes larger as the 
number of GPU nodes increases to 256: WAGMA-SGD achieves up to 1.37x speedup. The only algorithm with higher throughput than WAGMA-SGD is AD-PSGD, in which the asynchronous communication is completely overlapped with the computation.
These results show that WAGMA-SGD can better handle the unbalanced workload than the synchronous SGD algorithms (i.e., local SGD, Allreduce-SGD, D-PSGD, and SGP), as well as the bounded-staleness eager-SGD variant. In the latter case, while staleness is bounded, the algorithm still conducts a global collective communication for gradient averaging in each training iteration. In contrast, WAGMA-SGD keeps the collectives within each group, and thus has a better parallel scalability.

\begin{figure}[!h]
  \centering
    \centering\includegraphics[width=.82\linewidth]{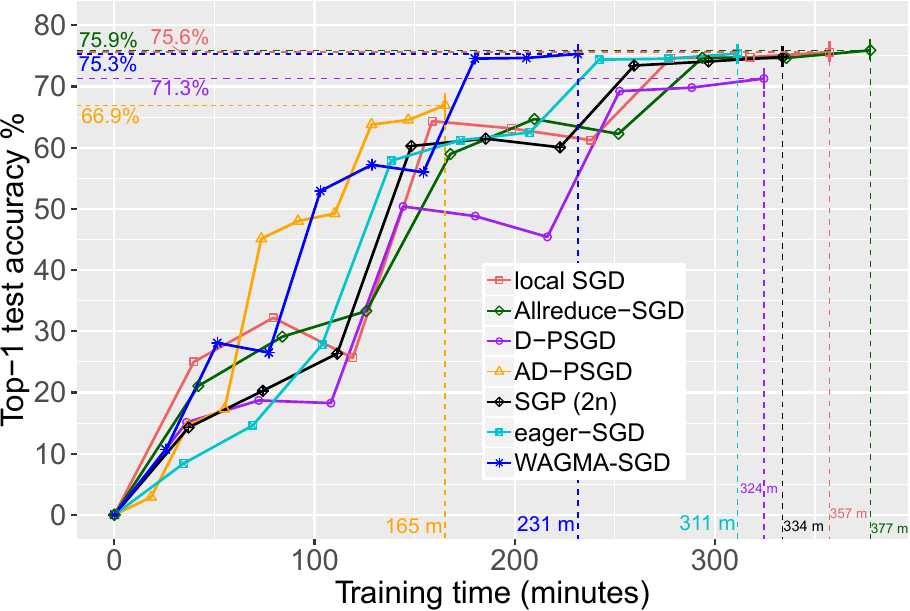}
  \caption{Top-1 validation accuracy of ResNet-50 on ImageNet training for 90 epochs using 64 GPU nodes. 
  Each point is at the boundary of every 10 epochs.} 
  \label{imagenettop1test}
\end{figure}

Fig.~\ref{imagenettop1test} presents the Top-1 validation accuracy when training for 90 epochs on 64 nodes, with a total batch size of 4,096 and simulated load imbalance. We can see that the accuracy of WAGMA-SGD (75.3\%) is very close to the standard Allreduce-SGD (75.9\%, in accordance with that in MLPerf~\cite{mattson2020mlperf}) and local SGD (75.6\%), and the slightly decreased accuracy is caused by the staled model parameters in the unbalanced workload environment; WAGMA-SGD significantly reduces the training time (e.g., 1.45x, 1.54x, and 1.63x speedup over SGP, local SGD, and Allreduce-SGD, respectively). Gossip-based SGD algorithms, such as D-PSGD and the higher-throughput AD-PSGD, attain much lower accuracy than the other variants. This can be explained by the fact that the algorithms have not fully converged, requiring more steps to be taken to achieve comparable accuracy~\cite{nadiradze2019popsgd}. We 
further train 130 epochs for AD-PSGD and linearly scale the learning rate, and the accuracy is improved to 70.4\%; however, it costs 238 minutes (longer than WAGAM-SGD which achieves a higher accuracy). Similarly, we improve the accuracy of D-PSGD to 74.4\% by doubling the communication steps in each iteration, which costs 410 minutes for 90 epochs. For SGP, we set and tune the number of communication neighbors to achieve the highest generalization using a directed exponential graph~\cite{assran2018stochastic}, which causes it to achieve higher accuracy (74.8\%) than D-PSGD and AD-PSGD, yet still lower than WAGMA-SGD. Note that the default setting for the number of communication neighbors in SGP is one, whereas we set it to two for better generalization performance. Overall, WAGMA-SGD achieves the highest accuracy-vs-time among all parallel SGD variants and quite robust to the unbalanced workloads.

We also train the model without simulated load imbalance for 90 epochs on 64 nodes, with a total batch size of 8,192. Local SGD, WAGMA-SGD, AD-PSGD, D-PSGD, and SGP achieve 1.52, 1.55, 1.61, 1.57, and 1.51 global\_step/s for the training throughput, respectively; and 75.3\%, 75.0\%, 67.1\%, 71.2\%, and 74.5\% for the Top-1 test accuracy, respectively. These results show that, when the workload is balanced, there is no big difference for the training throughput between the synchronous and asynchronous parallel SGD algorithms; WAGMA-SGD performs as well as the synchronous parallel SGD in terms of the model accuracy.

By setting the group size $S = \sqrt P = 8$, WAGMA-SGD has a faster model update propagation speed (globally propagate only using $\log_{S}P = 2$ iterations) than the gossip-based algorithms (globally propagate using at least $\log_{2}P = 6$ iterations), which makes WAGMA-SGD achieve higher accuracy. This is consistent with our analysis in Section~\ref{convergeproof}. 

\begin{figure}[t]
\centerline{\includegraphics[width=.82\linewidth]{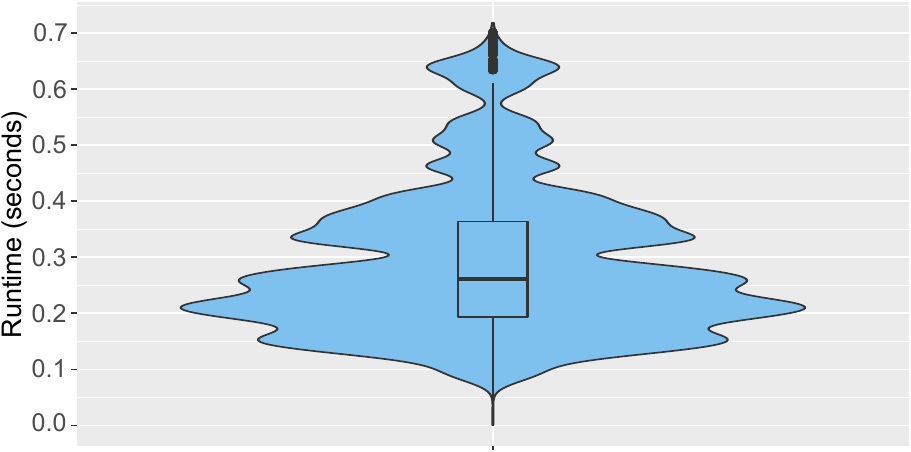}}
\caption{Runtime distribution of different sentences on a P100 GPU for a Transformer network on WMT17.}
\label{transimb}
\end{figure}

To analyze the convergence properties of WAGMA-SGD further, we conduct additional experiments. \one{} In the first experiment, we remove the wait-avoiding group collectives in WAGMA-SGD and only keep standard allreduce operations on the synchronization points, which is essentially equivalent to local SGD with a synchronization period $\tau = 10$. This causes the top-1 validation accuracy to sharply drop to 68.5\%. \two{} In a second experiment, we execute group model averaging without using the dynamic grouping strategy (i.e., the groups are fixed). In this case, the top-1 validation accuracy drops to 72.2\%. \three{} We also experiment with increasing the group size to 64 (i.e., a global collective). While accuracy does not increase, the throughput drops by factor of 1.24x. \four{} Lastly, we decrease the group size to 4 and observe that the top-1 validation accuracy drops to 72.8\%. 

The results from experiments \one{} and \two{} indicate that the combination of group allreduce operations and the dynamic grouping strategy is essential to achieve good generalization performance. The results from experiments \three{} and \four{} demonstrate that $S = \sqrt P$ empirically exhibits the best performance among different group size settings.

\subsection{Machine Translation}
\label{eval:transformer}

Transformers are sequence-to-sequence transducers that can be used to translate a sequence of words from one language to another.
We use the standard-sized Transformer network~\cite{vaswani2017attention}, which has 61,362,176 trainable parameters, to train English to German translation WMT17 dataset using TensorFlow as the basic platform.
While training the model, the computation overhead changes
with the length of the input and output sentences. The samples in the training dataset typically consist of sentences in various lengths and thus the training workload is unbalanced. As shown in Fig.~\ref{transimb}, even when using a bucketing strategy to group sentences with similar lengths, there is a high variance in workload size between samples. Specifically, in our experiment each local batch contains equal number of sentences sampled from a randomly selected bucket, where the maximum local batch size is set to 8,192 tokens. For WAGMA-SGD, we set the synchronization period $\tau = 8$ and the group size $S = \sqrt P$.

\begin{figure}[!ht]
\centerline{\includegraphics[width=.82\linewidth]{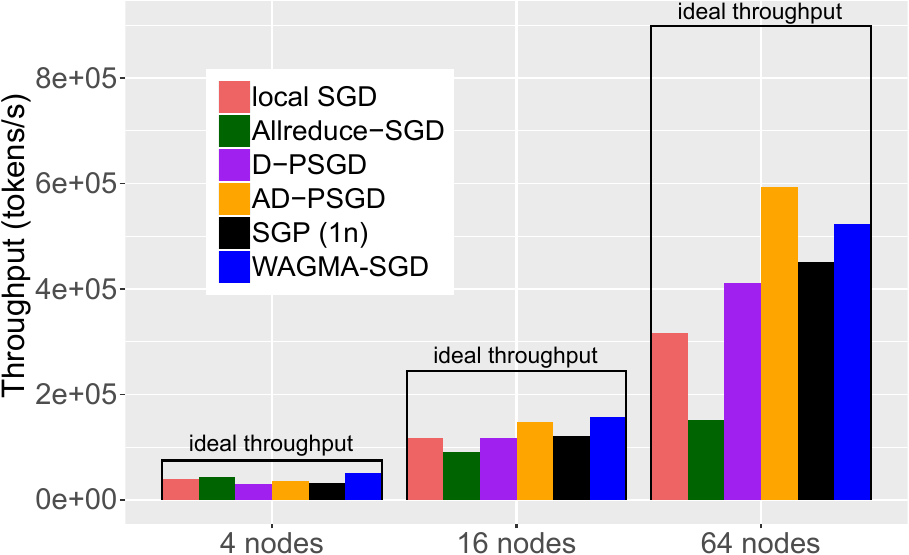}}
\caption{Throughput comparison between different parallel SGD algorithms for Transformer on WMT17.}
\label{transtp}
\end{figure}

Fig.~\ref{transtp} presents the training throughput as the number of GPU nodes increases from 4 to 64, where the top of the rectangle
indicates the ideal throughput without communication overhead.
On 16 GPU nodes, WAGMA-SGD achieves the highest throughput, compared with local SGD, Allreduce-SGD (implemented in Horovod~\cite{sergeev2018horovod}), D-PSGD, 
AD-PSGD, and SGP (one communication neighbor). When the number of GPU nodes increases to 64, as with image classification WAGMA-SGD exhibits a lower throughput than AD-PSGD but higher than the other variants. Observe that on 64 nodes, all algorithms perform far worse than the ideal throughput. We believe that this effect stems from the balance of the number of parameters (occupying 245 MB alone) vs. the operational intensity to compute backpropagation. Since transformer networks mostly consist of tensor contractions implemented as batched matrix products, which utilize GPUs well, communication overhead dominates and not even AD-PSGD manages to overlap communication with computation.

\begin{figure}[ht!]
\centerline{\includegraphics[width=.82\linewidth]{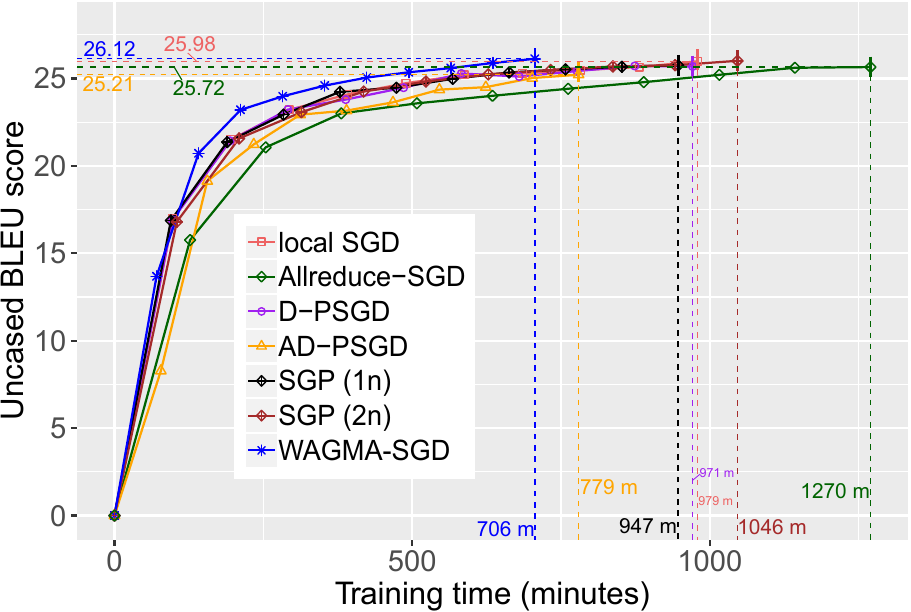}}
\caption{Uncased BLEU score for Transformer on WMT17 training for 10 epochs using 16 GPU nodes. 
Each point is at the boundary of one epoch.}
\label{transbleu}
\end{figure}

As for accuracy, Fig.~\ref{transbleu} presents the BiLingual Evaluation Understudy (BLEU) score (higher is better) on the test dataset after training for 10 epochs on 16 nodes. 
D-PSGD and AD-PSGD, have lower score (25.69 and 25.21, respectively) than the other SGD variants, likely because of
the slower model update propagation. SGP (1n, i.e., one communication neighbor) achieves higher score (i.e., 25.75) than D-PSGD and AD-PSGD. After increasing the number of communication neighbors to two in SGP (2n), the score increases to 26.01 (equivalent to local SGD, 25.98). However, this accuracy increase comes at the cost of reduced training speed compared with SGP (1n). Among all SGD variants, WAGMA-SGD not only achieves the highest score (i.e., 26.12, higher than 25.00 which is claimed in MLPerf~\cite{mattson2020mlperf}) but also uses the shortest training time (e.g., 1.48x, 1.10x and 1.39x speedup over SGP (2n), AD-PSGD, and local SGD, respectively).

We conduct additional experiments for WAGMA-SGD, similarly to Section \ref{sec:imagenet}: (1) Without using the dynamic grouping strategy (i.e., fixed groups), the score drops to 24.79; (2) By increasing the group size to 16 (i.e., global collective), accuracy does not improve and training throughput drops by a factor of 1.28x; and (3) By decreasing the group size to 2, the score drops to 24.53. These results reaffirm the conclusions from image classification.

\subsection{Deep Reinforcement Learning}
\label{eval:drl}

Due to the inherent characteristics of the problem, reinforcement learning poses a more challenging training process over supervised and semi-supervised learning. This also applies to the heterogeneity in workloads during training --- since the problems in question involve interacting with an environment in episodes (where failure terminates an episode early), a variety of episode lengths may occur within a single minibatch, in a way that cannot be anticipated or categorized in buckets.

We use the popular Proximal Policy Optimization (PPO) policy gradient  optimizer~\cite{wijmans2019decentralized} to train a model for robot navigation on a meta-dataset called Habitat~\cite{savva2019habitat}, which is composed of multiple heterogeneous environments. We first confirm previous claims~\cite{wijmans2019decentralized} and our own in Fig.~\ref{habitatimb}, where we collate the runtime distribution of 5,000 training iterations. The runtime is very widely distributed: from 1.7 seconds to 43.5, with a median below 2 seconds, which makes it an excellent use case for the load-rebalancing properties of WAGMA-SGD.

\begin{figure}[t]
\centerline{\includegraphics[width=.82\linewidth]{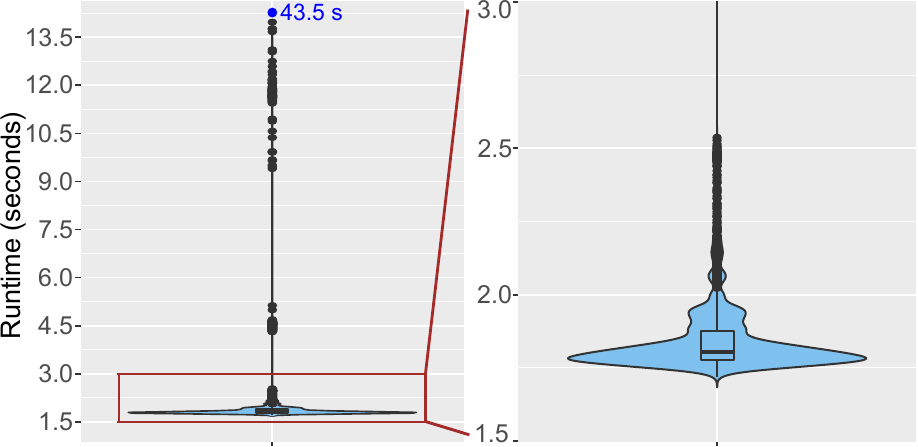}}
\caption{Runtime distribution of experience collecting on a P100 GPU in heterogeneous environments.}
\label{habitatimb}
\end{figure}

\begin{figure}[t]
\centerline{\includegraphics[width=.82\linewidth]{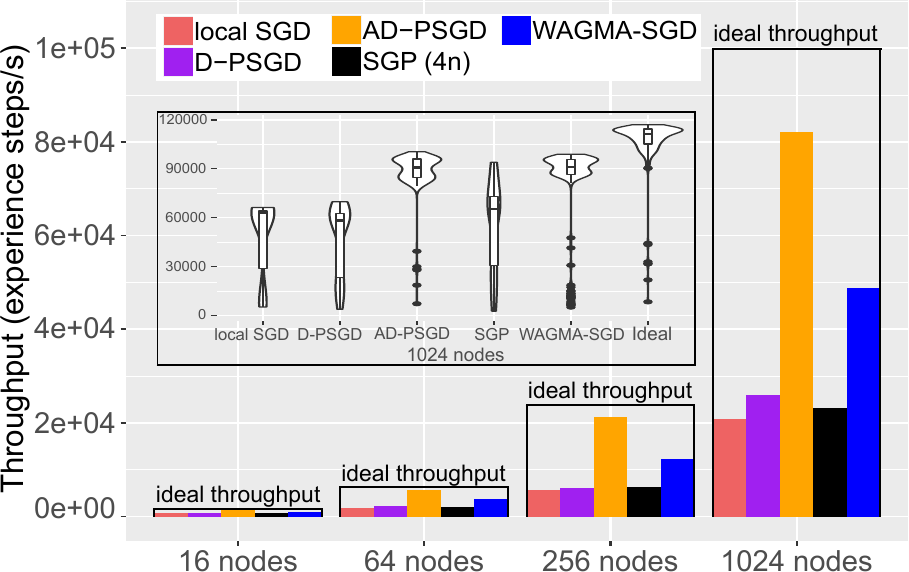}}
\caption{Throughput comparison between different parallel SGD algorithms for DDPPO on Habitat. (AD-PSGD achieves
the highest throughput but with a much lower score.)}
\label{ddppotp}
\end{figure}

\begin{figure}[t]
  \centering
    \centering\includegraphics[width=.82\linewidth]{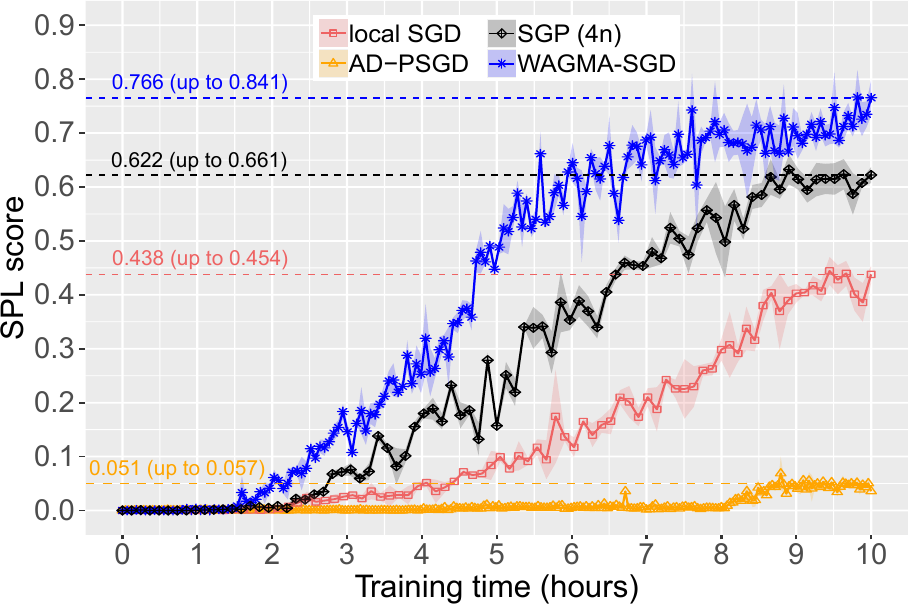}
  \caption{SPL score comparison for DDPPO on Habitat. Training on 64 GPU nodes for 10 hours. Each point is at the boundary of every 50 updates.} 
  \label{ddppospl}
\end{figure}

To evaluate the performance, we train a standard ResNet-LSTM model for navigation using PyTorch~\cite{paszke2019pytorch} as the basic platform. In particular, the network structure is composed of a ResNet-18 visual encoder, connected to a stack of two Long Short-Term Memory (LSTM)~\cite{hochreiter1997long} recurrent units functioning as the policy, containing 8,476,421 trainable parameters. The measured heterogeneous environments in Habitat, Gibson~\cite{xiazamirhe2018gibsonenv}
and Matterport3D~\cite{Matterport3D}, consist of interactive RGB-D datasets.
We set the experience steps to 128 and use the two vectorized (namely, optimized) environments, which means each GPU node needs to collect 256 experience steps for each training iteration. We set the WAGMA-SGD synchronization period to $\tau = 8$.

Fig.~\ref{ddppotp} presents the training throughput as the number of GPU nodes increases from 16 to 1,024, where the top of the rectangle indicates the ideal throughput without communication overhead.
Compared with local SGD, D-PSGD, and SGP (four communication neighbors) on 1,024 GPU nodes, WAGMA-SGD achieves 2.33x, 1.88x, and 2.10x speedup, respectively. The violin plot shows the throughput distribution. WAGMA-SGD only has lower throughput than AD-PSGD, since AD-PSGD is fully asynchronous. These results show that WAGMA-SGD excels in handling highly unbalanced workloads, achieving good scaling efficiency.

Complementary to training throughput, we study the Success weighted by Path Length (SPL) score (higher is better) after training the model for 10 hours on 64 GPUs. All models are tested four separate times to account for variability, and the average scores together with the standard deviation (shaded regions) over training time are plotted in Fig.~\ref{ddppospl}. As the figure shows, after training for 10 hours, AD-PSGD converges to a lower SPL score (on average 0.051), which does not rise over the course of 13,000 iterations. This is probably caused by the unbounded staleness and the overuse of the stale models in AD-PSGD. On the other hand, WAGMA-SGD achieves the highest score over time. SGP scores higher than local SGD, but not as well as WAGMA-SGD, whose quorum size is larger.

Beyond our experiments, the current state-of-the-art (SOTA) SPL score is 0.922~\cite{wijmans2019decentralized}, which is achieved after training on 2.5 billion experience steps. WAGMA-SGD consumes total 2.1 million experience steps after training for 10 hours using 64 GPUs, and achieves on average 83.1\% (up to 91.2\%) of the SOTA score. This indicates that WAGMA-SGD achieves generalization performance close to the SOTA using three orders of magnitude fewer iterations.

\section{Collectives in Context}
Collective operations have a core role in running applications efficiently at scale. As such, their optimization has led to several implementation and algorithmic variants~\cite{hoefler2007implementation,li2017cache,rabenseifner2004optimization,di2015exploiting}. \emph{Blocking} collectives~\cite{mpi-3.1} constitute the basic class of
operations. In this case, the collective call is allowed to return only when
the calling process has completed the actions needed to its participation in the operation. 
A first optimization to blocking collectives is to make them
\emph{non-blocking}~\cite{hoefler2007implementation}, enabling processes to return immediately and overlap other activities with the ongoing collective.

Some collectives require all processes to invoke the operation in order to complete, e.g., a reduction cannot be
computed before knowing all the values to reduce. Hence, their completion time can be influenced
by any skewing (imbalance) among the processes. \emph{Solo} collectives~\cite{di2015exploiting} remove this
synchronization overhead by making the collectives externally-triggerable: once a
process joins the collective, it sends an activation message to the other processes,
making them to start the collective independently from their state.
An issue of \emph{solo} collectives is that they make triggering the collective possible, even if there is only one process joining it. \emph{Majority} collectives~\cite{li2020taming} extend the \emph{solo} idea by
requiring that at least $P / 2$ processes join the collective before triggering
it. While these collectives are not guaranteed to be equivalent to their blocking
or non-blocking counterparts, they are suited for machine learning tasks, due
to the robustness of stochastic optimization to staleness.

Both \emph{solo} and \emph{majority} collectives aim to minimize the
synchronization overhead.  However, once activated, the collective is fully
performed, making the application pay the full operation cost plus the
activation overhead. \emph{Wait-avoiding group} collectives (this work) utilize the approach of solo collectives 
to achieve asynchrony, and reduce the overall operation cost further
by dynamically selecting subgroups of processes, each of which executing the
collective independently from the others. By applying the idea of wait-avoiding group allreduce to 
the decentralized parallel SGD based on model averaging, we propose WAGMA-SGD. 
Our work is orthogonal to communication compression or quantization methods, 
which reduce the communication volume by using less number of bits to represent the data (module update) to be transferred.

\section{Conclusion}

We show, both theoretically and in practice, that stochastic optimization via group model averaging --- asynchronously averaging the learned weights across subgroups of nodes --- functions well in large clusters. We prove that the algorithm converges under the standard conditions of SGD, and through a careful implementation of wait-avoiding collectives, we use the topology of the network to attain the best scaling results without losing accuracy. 
With the same number of steps, WAGMA-SGD achieves equivalent (or even higher) generalization scores as the standard synchronous SGD, while significantly reducing the training time (e.g., up to 1.48$\times$ speedup on Transformer) over the previous state-of-the-art, gossip-based SGD. Similar results are observed on the models from various sub-fields, where WAGMA-SGD consistently achieves the fastest time-to-solution. These results empirically prove that this approach successfully tackles the unbalanced training workloads in large scales, and brings asynchronous decentralized SGD to the regime of supercomputers.

\section*{Acknowledgment}
This project has received funding from the European Research Council (ERC) under the European Union's
Horizon 2020 programme (grant agreement DAPP, No. 678880; EPiGRAM-HS, No. 801039; and ERC Starting Grant ScaleML, No. 805223). 
T.B.N is supported by the Swiss National Science Foundation (Ambizione Project No. 185778).
N.D. is supported by the ETH Postdoctoral Fellowship. We also thank the Swiss National Supercomputing Center for providing the computing resources and technical support.

\bibliographystyle{IEEEtran}
\bibliography{mybib}

\newpage
\begin{IEEEbiography}[{\includegraphics[width=1in,height=1.25in,clip,keepaspectratio]{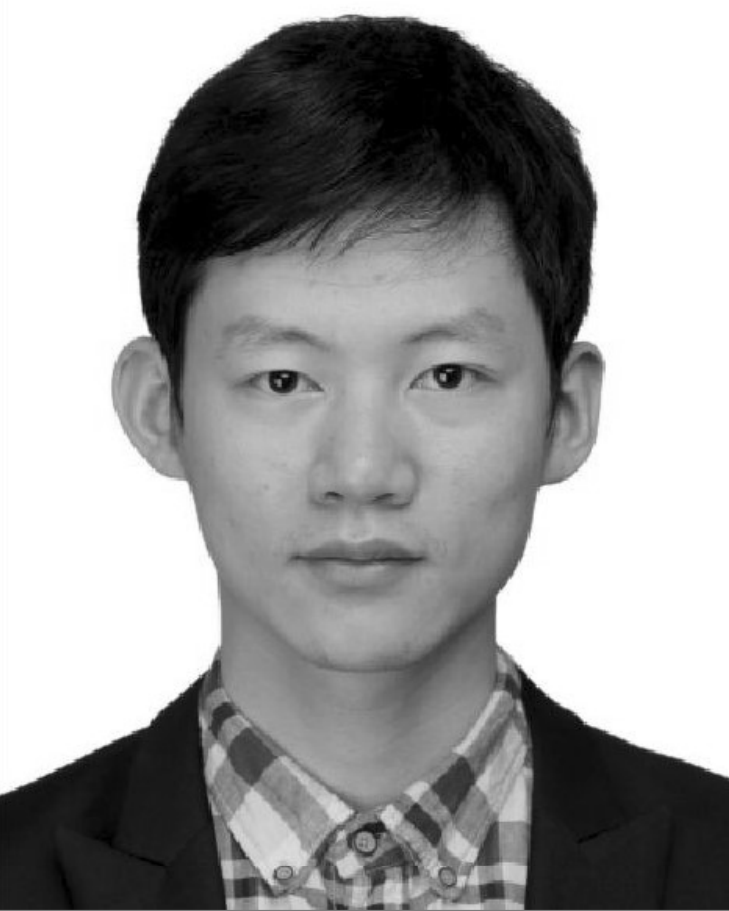}}]{Shigang Li} (Member, IEEE) is a postdoctoral researcher in Department of Computer Science, ETH Z\"{u}rich from August 2018 to now. He received the bachelor’s
degree in computer science and technology
and the PhD degree in computer architecture
from the University of Science and Technology
Beijing, China, in 2009 and 2014, respectively. He
was a joint PhD student at University of Illinois at
Urbana-Champaign from September 2011 to
September 2013 funded by CSC. He was an
Assistant Professor (from June 2014 to August
2018) with the State Key Lab of Computer Architecture,
Institute of Computing Technology, Chinese
Academy of Sciences. He
is a member of ACM. His research interests focus on parallel and distributed
computing, and parallel and distributed deep learning.

\end{IEEEbiography}

\begin{IEEEbiography}[{\includegraphics[width=1in,height=1.25in,clip,keepaspectratio]{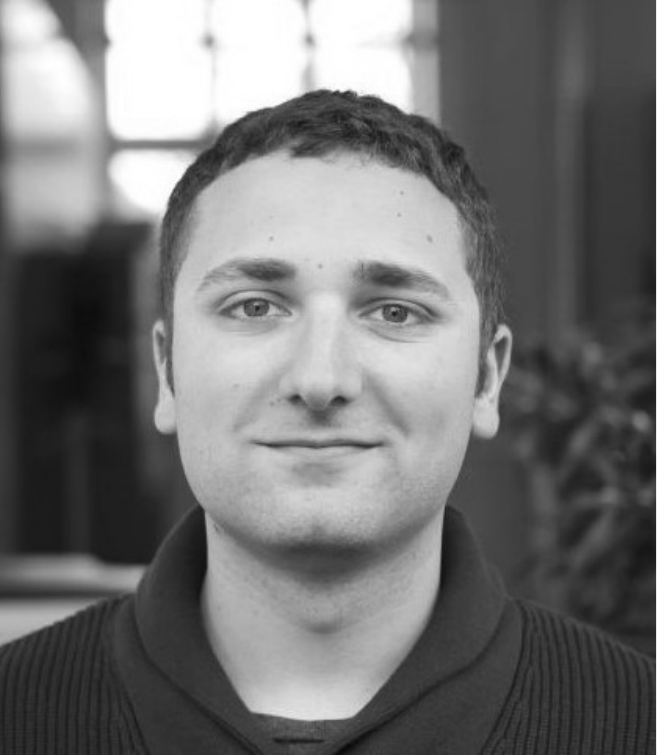}}]{Tal Ben-Nun} is a senior scientist in Department of Computer Science, ETH Z\"{u}rich. He received his PhD degree from the Hebrew University of Jerusalem in 2016, where he worked on programming models for distributed memory environments. His research interests include large-scale machine learning for scientific computing, machine comprehension of code via deep learning, and designing data-centric programming models for heterogeneous architectures.

\end{IEEEbiography}

\begin{IEEEbiography}[{\includegraphics[width=1in,height=1.25in,clip,keepaspectratio]{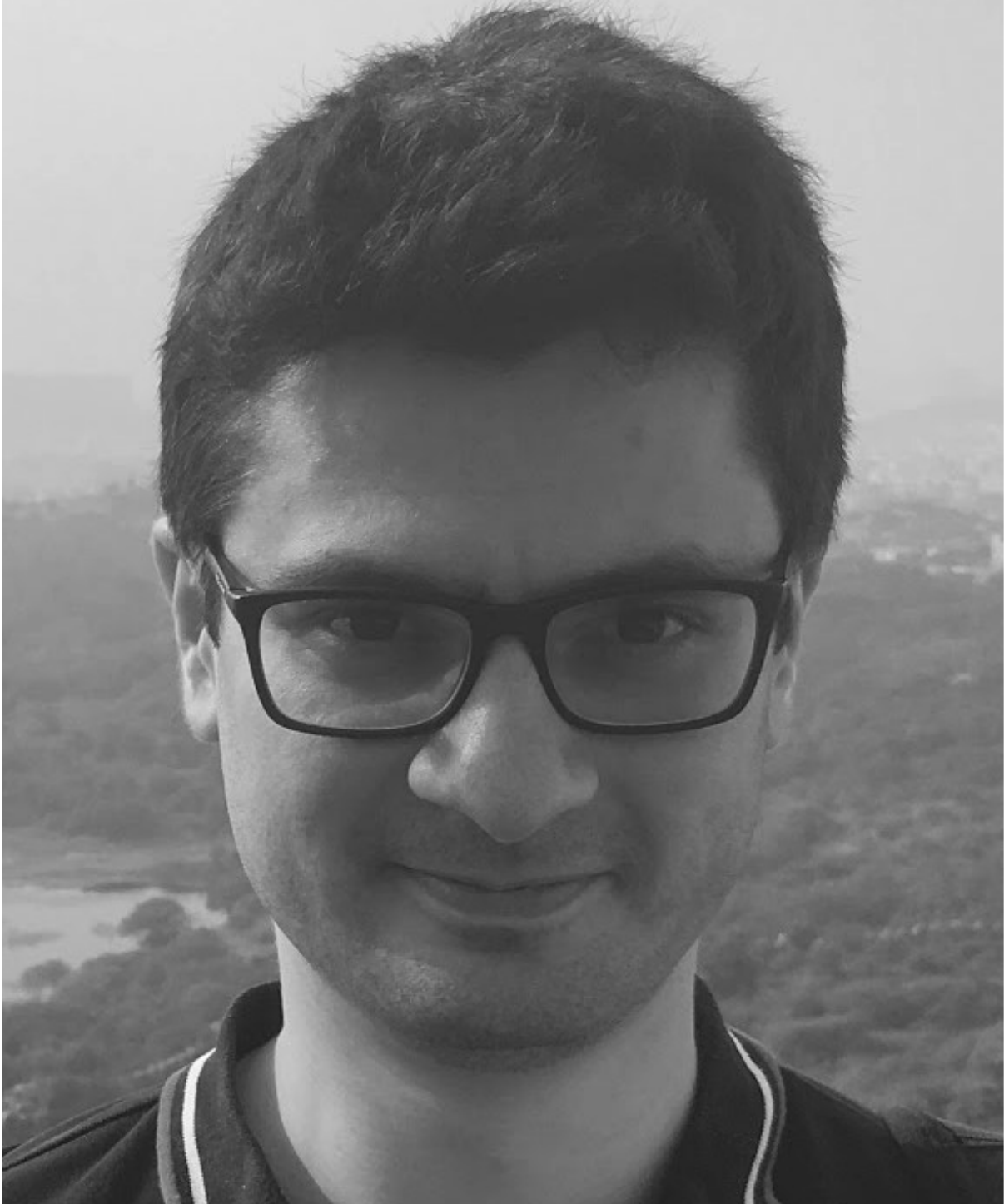}}]{Giorgi Nadiradze} is a PhD student at IST Austria. His research centers on fundamental problems in theoretical computer science, with an emphasis on dynamic load balancing processes and their applications. In particular, his work considers practical applications to concurrent data structures and distributed machine learning.

\end{IEEEbiography}

\begin{IEEEbiography}[{\includegraphics[width=1in,height=1.25in,clip,keepaspectratio]{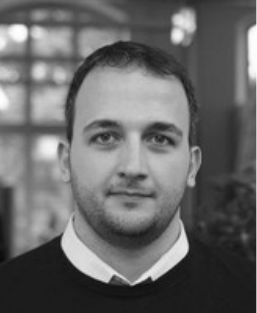}}]{Salvatore Di Girolamo}
is a PhD student at ETH Z\"{u}rich. He works on high-performance
networking with special focus on communication and computation network
offloading.

\end{IEEEbiography}

\begin{IEEEbiography}[{\includegraphics[width=1in,height=1.25in,clip,keepaspectratio]{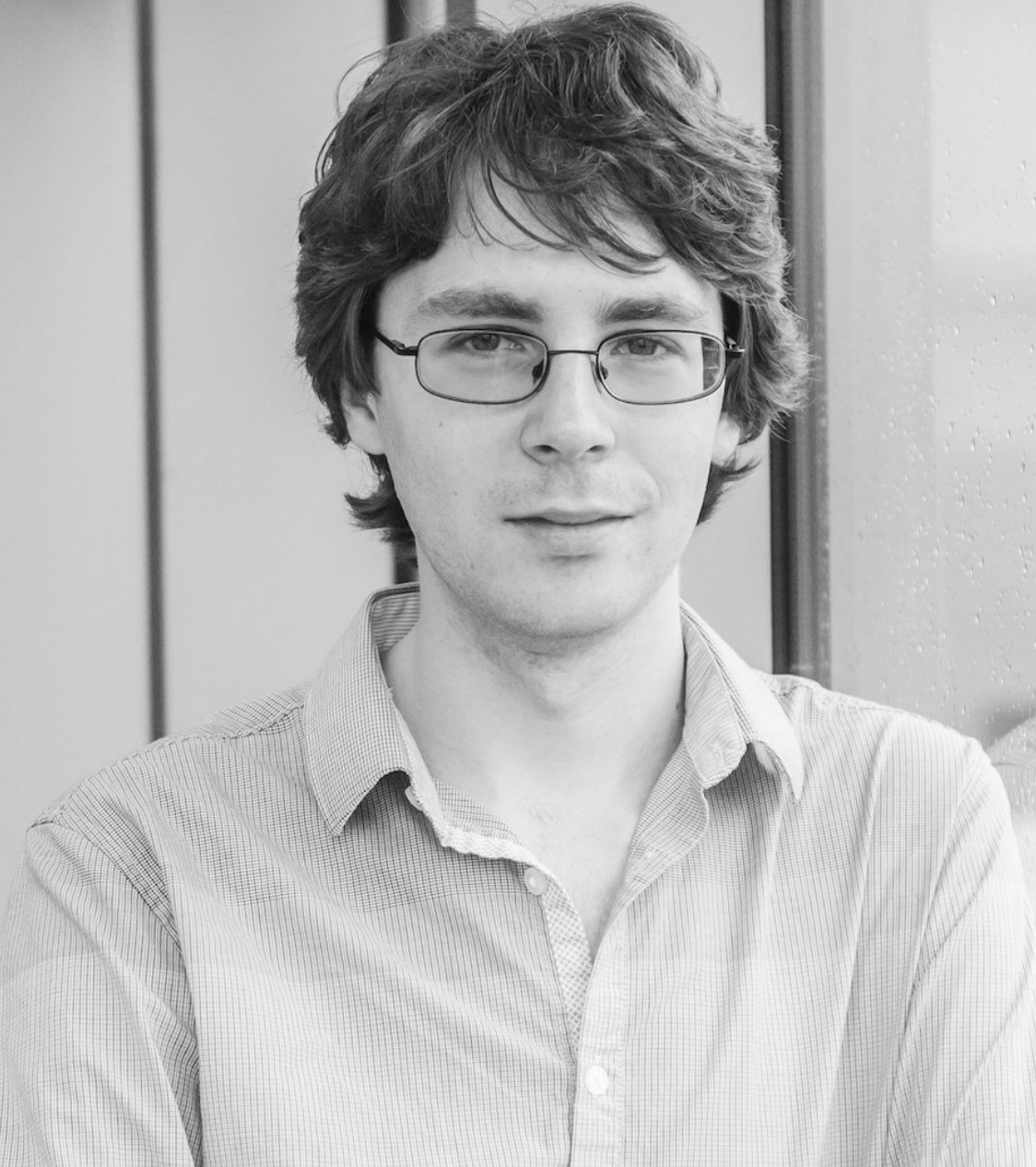}}]{Nikoli Dryden}
is an ETH Postdoctoral Fellow in the Scalable Parallel Computing Lab at ETH Z\"{u}rich. He received his PhD from the University of Illinois at Urbana-Champaign in 2019, where he worked on large-scale training of deep neural networks. His research interests include distributed machine learning, machine learning for computational science, parallel algorithms and runtimes, and communication and performance optimization.

\end{IEEEbiography}

\begin{IEEEbiography}[{\includegraphics[width=1in,height=1.25in,clip,keepaspectratio]{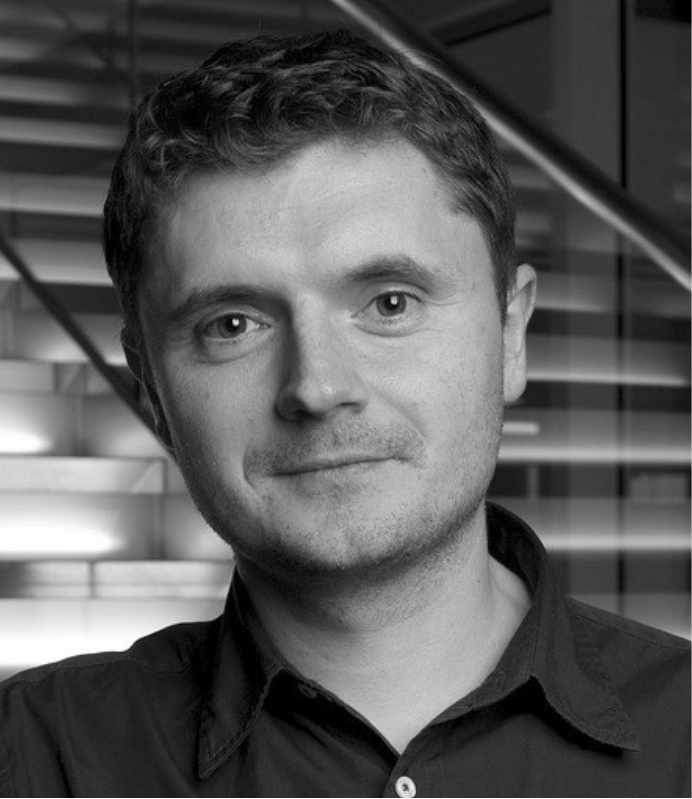}}]{Dan Alistarh} is an Assistant Professor at IST Austria, where he leads the Distributed Algorithms and Systems Group. His research focuses on concurrent data structures and distributed algorithms, and spans from algorithms and lower bounds, to practical implementations.
Before IST, he was a researcher at ETH Zurich and Microsoft Research, Cambridge, UK. Prior to that, he was a Postdoctoral Associate at MIT CSAIL.

\end{IEEEbiography}

\begin{IEEEbiography}[{\includegraphics[width=1in,height=1.25in,clip,keepaspectratio]{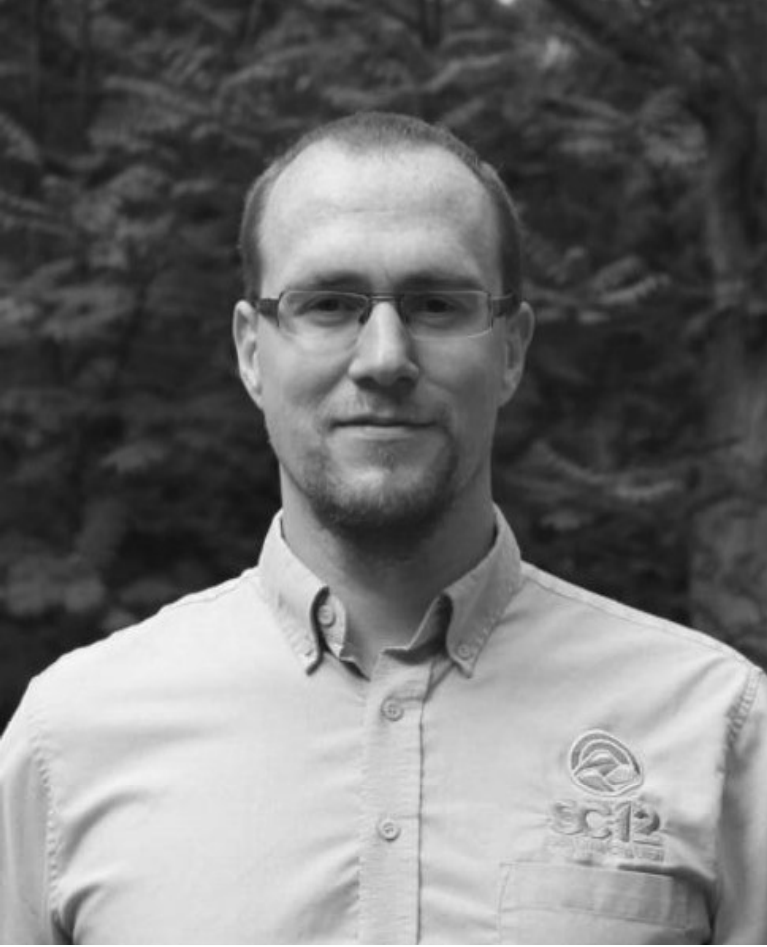}}]{Torsten Hoefler}
is a Full Professor of
Computer Science at ETH Z\"{u}rich, Switzerland. He directs the Scalable Parallel Computing Lab.
He is a key member of the MPI Forum where he chairs the Collective
Operations and Topologies working group.
His research interests revolve around the central
topic of Performance-Centric Software Development
and include scalable networks, parallel
programming techniques, and performance modeling.
\end{IEEEbiography}




\end{document}